\documentclass[a4paper,fleqn]{cas-dc}

\usepackage[numbers,sort&compress]{natbib}
\usepackage{longtable}
\usepackage{booktabs}
\usepackage{amsmath}
\allowdisplaybreaks
\usepackage{amssymb}
\usepackage{graphicx}
\usepackage{xcolor}
\usepackage{algorithm}
\usepackage{algpseudocode}
\usepackage{multirow}
\usepackage{enumitem}
\usepackage{makecell}
\usepackage{colortbl}  

\newcommand{\nollm}{^\text{noLLM}}

\usepackage{cuted}
\usepackage{caption}
\usepackage{xurl}

\usepackage{amsthm}
\theoremstyle{plain} 
\newtheorem{proposition}{Proposition}
\newtheorem{lemma}{Lemma}

\usepackage{mathtools}
\newcommand{\keepcomment}{1} 
\usepackage{soul}
\ifnum\keepcomment=1
	\usepackage[colorinlistoftodos,textsize=scriptsize]{todonotes} 
    \newcommand{\stkout}[1]{\ifmmode\text{\st{\ensuremath{#1}}}\else\st{#1}\fi}
    
\else
	
	\usepackage[disable]{todonotes} 
\fi

\ifnum\keepcomment=1
	\usepackage[colorinlistoftodos,textsize=scriptsize]{todonotes} 
    \newcommand{\stkoutz}[1]{\ifmmode\text{\st{\ensuremath{#1}}}\else\st{#1}\fi}
    
\else
	
	\usepackage[disable]{todonotes}
\fi

\begin{document}

\shorttitle{From Hype to ROI: An Economic Model for LLM Adoption} 

\title[mode=title]{From Hype to ROI: An Economic Model for LLM Adoption in Information Systems and Security Operations}

\shortauthors{Casoria et~al.}

\author[1]{Pellegrino Casoria}%
\cormark[1]                         
\ead{pellegrino.casoria@unina.it} 
\credit{Conceptualization, Methodology, Software, Validation, Formal analysis, Investigation, Resources, Data Curation, Writing - Original Draft, Writing - Review \& Editing, Visualization, Project administration}  

\author[2]{Andrea Araldo}%
\credit{Conceptualization, Methodology, Validation, Formal analysis, Investigation, Data Curation, Writing - Original Draft, Writing - Review \& Editing, Visualization, Supervision}

\author[2]{Rosario Patanè}%
\credit{Software, Formal analysis, Writing - Review \& Editing, Visualization}

\author[1]{Simon Pietro Romano}%
\credit{Writing - Review \& Editing, Supervision}

\cortext[1]{Corresponding author}

\affiliation[1]{organization={University of Naples Federico II},
                department={Department of Electrical Engineering and Information Technology (DIETI),},
                city={Naples},
                country={Italy}}

\affiliation[2]{organization={SAMOVAR, Télécom SudParis, Institut Polytechnique de Paris},
                city={Palaiseau},
                country={France}}

\begin{abstract}
Large Language Models~(LLMs) promise significant efficiency gains across information-intensive operations by automating repetitive language-based tasks. However, adoption decisions remain largely opportunistic, as organizations struggle to move from pilots to sustainable industrialization. Without systematic economic frameworks, practitioners cannot assess whether LLM investments create value in a given context.

This paper develops a general economic model to support informed LLM adoption decisions in operational information systems. Facing high uncertainty, we represent parameters as random variables centered on authoritative benchmarks, and evaluate the model by Monte Carlo simulation across pessimistic-to-optimistic assumptions on LLM efficiency. The ROI distributions show the probability of profitability and the Value at Risk for each configuration, together with a sensitivity analysis of its main drivers.

As a first instantiation, we apply the model to Security Operations Centers~(SOCs), mapping workflows onto established standards to identify tasks suited to LLM augmentation. The model balances productivity gains, driven by alert volume, analyst effort, and LLM efficiency, against operational costs, including inference, infrastructure, and error remediation, covering both cloud API and self-hosted deployments.

Results show that adoption economics are context-dependent, and that automation is an exploitable lever, gated by accuracy, while the cost of errors and the model's price are secondary. Across the analyzed models and prices, most configurations are unprofitable, with profitability limited to high-accuracy, low-uncertainty deployments. 

The model gives organizations a quantitative basis for evidence-based adoption decisions, by locating the boundary between profit and loss, and the loss risked beyond it.
\end{abstract}


\maketitle

\section{Introduction}

Large Language Models (LLMs) have emerged as a candidate technology for automating knowledge-intensive work across information systems. Gartner reports that 42\% of organizations are piloting AI agents, while 46\% plan to adopt them in the near term \cite{Gartner2025}. However, translating pilots into production value remains difficult. MIT research indicates that 95\% of enterprise generative AI implementations fail \cite{MIT2025}. S\&P Global found that 42\% of companies abandoned most AI initiatives in 2025, compared with 17\% in 2024, and that organizations scrapped 46\% of proof-of-concepts before production \cite{SPGlobal2025}. McKinsey reports a similar pattern: although 65\% of organizations use generative AI, only 10--15\% report measurable business impact \cite{McKinsey2024}. Gartner estimates that 85\% of AI initiatives never reach production \cite{GartnerAI2024}. 

Recent research emphasizes that technology investment decisions are economic as well as technical. They involve trade-offs in timing, incentives, and risk, and require quantitative approaches that explicitly account for these factors \cite{woods2025cyberrisk, sognen2025delay}. The World Economic Forum echoes this view, noting that leaders must quantify risks and their economic impacts to align investments with core business objectives \cite{WEF2025}. Yet practitioners lack quantitative decision-support tools for assessing whether LLM adoption will generate returns in a specific operational context. This shortfall mirrors the gap between adoption intent and realized value.
The OECD identifies difficulty estimating ROI in advance as the most frequently reported obstacle to AI adoption \cite{OECD2025AIAdoption}. This paper addresses that challenge.

A prominent, cost-intensive example of this challenge is the Security Operations Center (SOC). The threat landscape is expanding in both the volume and speed of attacks, straining SOC teams beyond their staffing capacity and pushing them toward automation. To illustrate the scale, ENISA recorded 4,875 incidents between July 2024 and June 2025 \cite{ENISA2025}, while Microsoft reports 600 million daily attacks across its customer base and more than 1,500 unique threat groups \cite{Microsoft2024}. Attack dynamics have also shifted. CrowdStrike observed average breakout times of 48 minutes in 2024 and found that 79\% of detected intrusions were malware-free, relying instead on credential abuse and hands-on-keyboard techniques that evade signature-based defenses \cite{CrowdStrike2025}.

The workforce needed to meet this demand is not keeping pace. The global cybersecurity workforce gap has reached 4.8 million, a 19\% year-over-year increase; in addition, 67\% of organizations report staffing shortages and 90\% identify skills gaps \cite{ISC2024}. Retention compounds the problem: 70\% of analysts with fewer than five years of experience leave within three years \cite{SANS2025}. Fewer than 15\% of organizations report confidence in both their staffing levels and skills coverage \cite{WEF2025}. These capacity constraints undermine operational effectiveness. Understaffed teams struggle to triage event volumes, while experienced analysts command increasing labor-market premiums. 

Against this backdrop, the Cognitive SOC has emerged as a natural setting for LLM augmentation. It uses AI/ML and cognitive computing to emulate human-like reasoning across a continuous detect--analyze--respond cycle \cite{binbeshr2025cognitive}. This context motivates decision-support models for LLM adoption under deep uncertainty. We instantiate our model on the SOC as a first, high-stakes operational setting.

\subsection{Contributions of this work}
This paper develops a general economic model for evaluating LLM adoption in operational information systems under high uncertainty. It determines \emph{whether} adoption is profitable and under what conditions, and identifies \emph{where} organizations should focus their efforts to improve profitability. Taking the SOC as a first instantiation, we make three contributions:
\begin{enumerate}[leftmargin=*]
    \item \textbf{A general cost-benefit model of LLM adoption.} We weigh productivity gains against operational costs, including the time operators spend validating LLM outputs and redoing tasks the model gets wrong. The model also accounts for the risks associated with uncertainty in LLM effectiveness. We derive closed-form conditions for the maximum capital expenditure an LLM deployment can bear while still creating value.
    \item \textbf{Benchmark-grounded parameterization.} We instantiate the model using parameters from authoritative sources, including salaries, task effort, and work-item volumes, as well as LLM cost and accuracy. We organize deployment options into archetypes that span cost, accuracy, and deployment characteristics, providing a reusable basis for evaluation.
    \item \textbf{Decision support for investment prioritization.} Using Monte Carlo simulation at different uncertainty levels, we estimate each configuration's probability of profitability and value at risk. Sensitivity analysis identifies the parameters that most strongly influence outcomes, helping organizations prioritize efforts to achieve positive returns. We release the implementation as open source so practitioners can analyze their own operating conditions.
\end{enumerate}

These results show that, given today's models, prices, and SOC parameters, LLM adoption is unprofitable for most configurations. Adoption is profitable only for high-accuracy API deployments with low capital costs and low uncertainty. Profitability depends mainly on model autonomy, provided the model is accurate enough for its outputs to be trusted. Under low uncertainty, full autonomy raises the capital expenditure that the most accurate models can support with high confidence from about 10\% to 80--90\% of baseline operating costs. By contrast, eliminating the rework penalty raises this threshold only to 20\%. In self-hosted deployments, GPU costs outweigh the efficiency savings at typical SOC alert volumes.

\section{Background and Related Work}

A substantial body of research examines the economics of technology adoption and the modeling of information-intensive operations. A rapidly growing body of work also studies the integration of LLMs into professional workflows. However, the economic implications of LLM adoption remain largely unexplored. Existing studies characterize capabilities, benchmark accuracy, and propose deployment taxonomies, but do not provide a quantitative framework for evaluating adoption's economic value. We review two complementary strands that motivate our contribution: the methodological and economic foundations of our modeling choices (Sec.~\ref{sec:methodological}), and empirical work on the economics and productivity of security operations, the domain in which we apply our model (Sec.~\ref{sec:soc-economics}).

\subsection{Methodological and Economic Foundations} \label{sec:methodological}
Treating technology adoption as an investment under uncertainty is well established. The economics of information security has matured into a cross-disciplinary field~\cite{Anderson2006Economics}, which has shown that incentive misalignment causes failure as often as technical flaws. Protection is framed as an optimization problem with diminishing marginal returns, where optimal spending balances cost against expected loss~\cite{Gordon2002Economics}. We adopt this cost-benefit logic but, whereas prior work offers conceptual frameworks and deterministic analysis, we operationalize it through distributional reasoning and simulation. Our model also includes costs specific to LLMs: inference and infrastructure, output validation, and error remediation. AI's economic effect can be derived at the task level from the fraction of tasks affected and the savings on each~\cite{acemoglu_simple_2025}. We therefore model the impact of LLM adoption as the sum of time and cost savings across tasks, holding task composition fixed and final output quality at the human baseline.

Decision analysis and Monte Carlo methods provide a basis for reasoning about investments under incomplete information. They represent outcomes as distributions rather than single point estimates~\cite{Clemen2013MakingHard, Zhang2021MonteCarlo, Mateos2006MonteCarlo}. We build on this tradition to characterize the ROI distribution across plausible ranges of LLM efficiency and cost. Simulation is also a recognized source of knowledge in security economics, provided its assumptions are explicit~\cite{Kianpour2025Simulations}. It lets us track interactions among uncertain parameters while abstracting the system to its economically relevant properties.

A separate line of research examines the value of AI in knowledge work. It identifies accuracy and the cost of correcting errors, rather than model price, as decisive factors. AI has a jagged frontier of capability, producing gains on self-contained tasks but losses on complex managerial tasks that require context-specific judgment~\cite{dellacqua_navigating_2026}. As automation costs fall, verification, limited by human time and feedback latency, becomes the binding constraint~\cite{catalini_simple_2026}. These findings motivate our modeling of accuracy, which determines the effort saved, as well as the verification and rework demanded by errors. Full automation is often not cost-minimizing. Because accuracy has diminishing returns, retaining humans for residual tasks emerges as the economic equilibrium~\cite{li_economics_2026}. At the individual-model level, the expected cost of a correct answer combines accuracy and inference cost. Researchers have proposed it as a measure of language-model productivity and benchmarked it against the cost of a human expert~\cite{erol_cost_pass_2026}. However, prior work does not extend these insights from individual tasks and models to the organization. Our model combines accuracy with the cost of verifying and reworking errors to assess, under uncertainty, the value of adoption in a specific operational setting.

\subsection{SOC Economics and Productivity} \label{sec:soc-economics}
The economics and productivity of security operations have been studied for over a decade. This literature provides the empirical grounding for our parameters.

A study of GenAI adoption across hundreds of organizations offers early empirical evidence of productivity gains in security operations, although confounders prevent causal identification~\cite{bono_generative_2025}. This motivates an ex-ante economic assessment of adoption.

Embedded anthropological studies show that profit-based metrics cannot capture all SOC value and that incident-processing time varies by threat type~\cite{sundaramurthy_tale_2014}. They trace the difficulty of measuring value to a mismatch between what analysts value and what management measures~\cite{sundaramurthy_turning_2016}. These findings motivate our cost-benefit approach, our severity-dependent effort model, and our inclusion of costs beyond throughput. This broader accounting is important because management's reliance on automation contrasts with analysts' skepticism about the metrics used to justify it~\cite{kokulu2019}.

Work on machine-learning network intrusion detection makes a closely related argument: tools should be judged by operational value, weighing effectiveness against lifecycle cost rather than benchmark performance. It also calls for reporting the factors that determine value before deployment~\cite{apruzzese_sok_2023}. However, it addresses a single detection model, not adoption at the organizational level.

Research on LLMs in SOC operations is substantial, spanning systematic reviews, capability-maturity taxonomies, and qualitative studies of analyst use~\cite{habibzadeh_large_2026, srinivas_ai-augmented_2025, singh_llms_2025}. These studies identify alert fatigue, workforce shortages, and the gap between pilots and deployment as unresolved problems. They consistently find that LLMs serve as human-supervised sensemaking aids rather than autonomous decision-makers, with deployments remaining at low levels of autonomy. Yet evidence relevant to adoption decisions remains fragmentary. For incident summarization, LLMs lack the security reasoning required for autonomous operation and may omit critical details or introduce factual errors. Collaborative use, by contrast, can reduce analyst effort and improve consistency~\cite{kramer2025integrating}. Other work formalizes human--AI teaming modes but does not establish which mode creates net economic value~\cite{chhetri2024human}. Reported cost savings remain limited to indirect proxies~\cite{binbeshr2025cognitive}.
 
However, this literature does not model adoption costs, return on investment, or the tradeoff between analysts and automation. Prior contributions are descriptive, empirical, or conceptual, focus on post-deployment measurement, or target non-LLM technologies; none provides a decision-support tool for evaluating adoption. We address this gap with a quantitative cost-benefit model and Monte Carlo simulation. The model assesses, before adoption, whether LLMs create net value for a given SOC configuration and under what conditions. It explicitly represents LLM-specific factors, including accuracy, token and infrastructure costs, and error remediation.

\subsection{A Reference SOC Model} \label{sec:reference-soc}

A \textit{Security Operations Center (SOC)} is the organizational unit responsible for monitoring and detecting security events. Incident response is, in principle, delegated to a \textit{Computer Security Incident Response Team (CSIRT)}, though in most organizations the two functions are combined~\cite{enisa_how_2020}. 

Modeling the economic impact of LLM on SOC requires characterizing the tasks that compose its work. Several standards codify these operational functions, each emphasizing a different aspect, among them \textit{FIRST's CSIRT Services Framework}~\cite{first_csirt_2019}, \textit{NIST CSF~2.0}~\cite{nist_csf_2024}, \textit{ISO/IEC~27035}~\cite{isoiec_27035_2023}, and \textit{IETF RFC~2350}~\cite{ietf_rfc_1998}. We build our reference capability model on \textit{NIST SP~800-61}, a widely adopted reference for incident response (Rev.~3~\cite{nist_sp_r3_2025} for the overall structure and Rev.~2~\cite{nist_sp_r2_2012} for the detailed breakdown of actions), as it decomposes incident handling at the granularity of analyst work, the unit our cost model requires.

\subsubsection{Tasks}

\begin{figure}
    \centering
    \includegraphics[width=0.95\linewidth]{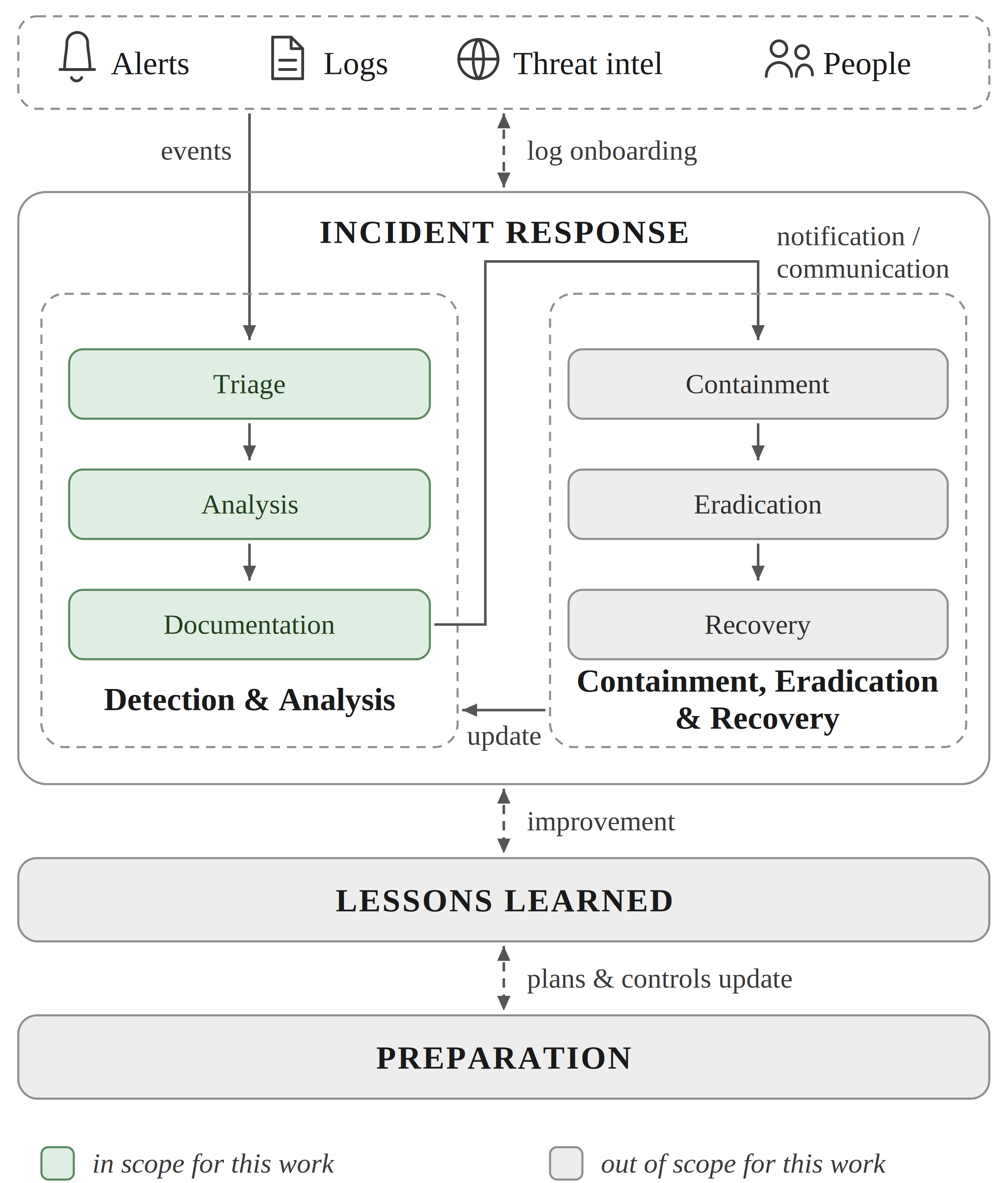}
    \caption{SOC incident response lifecycle}
    \label{fig:soc-lifecycle}
\end{figure}

Figure~\ref{fig:soc-lifecycle} depicts our reference incident response lifecycle. Integrated sources (alerts, logs, threat intelligence, people) generate a continuous stream of events; analysts triage and validate them to filter out noise, then analyze true positives to establish scope and cause, and document their findings. Confirmed incidents then move to containment, eradication and recovery, whose findings loop back into analysis when new evidence surfaces~\cite{nist_sp_r2_2012}.

In this paper's simulations, the model captures recurring operational cost only. Specifically, we consider \textit{Detection and Analysis} tasks, excluding the execution of containment, eradication and recovery, whose effort is site-specific, involves teams beyond the SOC, and can span weeks for larger incidents. It is therefore not a recurring unit of work with an estimable per-incident cost. The selection of response measures, however, is treated as the closing step of analysis and thus remains in scope.

Triage also assigns a severity level to each incident, which determines possible escalation: our simulations consider only low-, medium- and high-severity incidents, since critical ones are redirected to a dedicated response team drawing on roles outside the SOC. We treat the escalated fraction as a model parameter, as NIST prescribes no escalation thresholds. Evidence handling and communication beyond the initial notification are cross-cutting and left out, as are preparation and lessons learned.

\subsubsection{Professional Roles}

Several frameworks catalog SOC roles and competences. The \emph{NIST NICE Workforce Framework}~\cite{nist_nice_2026} describes work through task, knowledge and skill statements grouped into work roles, and the \emph{ENISA European Cybersecurity Skills Framework (ECSF)}~\cite{enisa_ecsf_2022} defines twelve role profiles with associated tasks, skills and competences; FIRST's \emph{CSIRT Services Roles and Competencies} addendum maps operational roles onto service functions, drawing competences from NICE~\cite{first_csirt_roles_2019}.

These frameworks map incident-management work onto a small number of coarse-grained roles, too aggregated for a cost model that distinguishes who performs each task. We therefore combine two elements. From ITIL (Information Technology Infrastructure Library)~\cite{axelos_itil_2014} we take the distinction between \emph{functional escalation}, which routes work to whoever has the required skills or system knowledge, and \emph{hierarchical escalation}, which routes it upward for seniority or authority. From practice we take the organization of SOC analysts into levels (L1--L3), a convention found in practitioner literature, vendor documentation and job descriptions, but not prescribed by any formal standard. Table~\ref{tab:soc-levels} states the operational definitions we adopt, with functional escalation driving the L1$\to$L2$\to$L3 progression on task complexity and hierarchical escalation setting each level's seniority (i.e.\ typical years of experience).

\begin{table}
    \centering
    \footnotesize
    \caption{SOC Roles by functions and seniority}
    \label{tab:soc-levels}
    \begin{tabular}{l >{\raggedright\arraybackslash}p{4cm} >{\raggedright\arraybackslash}p{2.5cm}}
        \toprule
        \textbf{Level} & \textbf{Functions} & \textbf{Experience} \\
        \midrule
        L1 & Procedural execution via documented playbooks & Entry-level (0--2 years) \\
        L2 & Analytical judgment and cross-system correlation & Mid-level (3--5 years) \\
        L3 & Deep technical analysis and specialized expertise & Senior / SME (6+ years) \\
        \bottomrule
    \end{tabular}
\end{table}

\subsubsection{Performance Metrics}

Measuring SOC performance is an open problem. Academic work has proposed metric sets and analyst-performance frameworks~\cite{agyepong_challenges_2020, agyepong_systematic_2023}, yet systematic studies find the established metrics insufficient and not systematically developed, being loosely derived from generic security metrics and industry best practice~\cite{vielberth_security_2020, forsberg_technical_2023}. No single standard therefore governs which metrics a SOC should adopt.

For a consistent operational vocabulary we rely on \emph{SOC-CMM}~\cite{soc-cmm_metrics_2024}, a widely used practitioner maturity framework, which organizes a set of 100 metrics by \emph{domain}, \emph{category} and \emph{type} (count, percentage, time, cost, trend). With reference to the operating model of Fig.~\ref{fig:soc-lifecycle}, we retain the \emph{Security monitoring} and \emph{Security incident management} categories (domain: \emph{Services}). Our economic model draws on this foundation to derive its core variables (alert volume, accuracy, severity, per-incident cost, per-task effort), grounding them in established operational practice. Among its time metrics are the \emph{Mean Time to Pick-up} (MTTP), from alert detection to the start of triage; the \emph{Mean Time to Analyse} (MTTA), from the start to the completion of alert analysis; and the \emph{Mean Time to Respond} (MTTR), from incident identification to the hand-over to incident response.

\section{Methodology: The Cost Model} \label{sec:general-cost-model}

This section formalizes our cost model. To clarify the presentation, we first keep the formulation general and do not distinguish between scenarios with and without LLM adoption. After establishing this general formulation, we mark baseline quantities without LLMs as ``$\square^{\nollm}$'' and quantities under a specific LLM model $m\in\mathcal M$ as ``$\square^{m}$''.

An information system receives a stream of operational events that must be handled. Events may belong to different classes, each with its own criticality level; let~$\mathcal A$ denote the set of criticality levels. In a SOC, these events are typically alerts with different severity levels. Each event requires one or more tasks; let~$\mathcal K$ denote the set of task types. Let $\mathcal R$ denote the set of roles in the information system.

Let $\Delta T$ be the analysis window, measured in hours. Let $\tau_{k,r,a}$ be the time, in hours, that an operator with role $r$ spends on task $k\in\mathcal K$ for an event of criticality $a\in\mathcal A$.
Let $c_r>0$ be the hourly cost to the organization of an employee with role $r\in\mathcal R$. The values used in this study are described in Section~\ref{sec:par_calibration}.

\subsection{Cost without LLM}
Let $C_\text{inv}$ be the portion of operational cost per unit of time that is not directly related to operations. It includes costs such as facilities staff, building maintenance, and infrastructure. This cost does not change materially with the introduction of an LLM. The personnel time spent per unit of time by operators with role $r$ is
\begin{align}
\tau_r\nollm = 
N\cdot
\sum_{a\in\mathcal A}
h_a \cdot p_a\cdot 
\sum_{k\in\mathcal K}
\tau_{k,r,a}
\label{eq:cost-no-LLM}
\end{align}
where $N$ is the number of events arriving during the analysis window $\Delta T$. A fraction $p_a\in[0,1]$ of these events has criticality $a\in\mathcal A$, with $\sum_{a\in\mathcal A} p_a=1$. Only a fraction $h_a\in[0,1]$ of events with criticality $a$ requires operator attention. The remaining fraction, $1-h_a$, is handled automatically by traditional frameworks without substantial human intervention.

The personnel cost per unit of time of an organization that does not adopt LLMs is:
\begin{align}
\label{eq:Cnollm}
C\nollm  = 
\underbrace{
\sum_{r\in\mathcal R} c_r \cdot \tau_r\nollm
}_{C_\textit{pers}\nollm}
+ C_\text{inv}
\end{align}
where $C_\textit{pers}\nollm$ denotes the operational personnel cost per unit of time when LLMs are not used.

The time spent by an operator with role $r$ on a task $k$ for an event of criticality $a$ is
\begin{align}
\tau_{k,r,a} = \tau_{k,a}\cdot \alpha_{a,r}
\end{align}
where $\tau_{k,a}$ is the time required to perform task $k$ for a single event of criticality $a$, and $\tau_a=\sum_k\tau_{k,a}$ is the total analyst effort per event of that criticality. A fraction $\alpha_{a,r}$ of this time is spent by an operator with role $r$, where $\sum_{r}\alpha_{a,r}=1$ for all $a\in\mathcal A$.

The values of $\tau_{k,a}$ and $\alpha_{a,r}$ used in the numerical results are reported in Sec.~\ref{sec:task_effort} (Tab.~\ref{tab:effort-alloc}).

\subsection{Cost when using LLMs}
Let $\mathcal M$ denote the set of available LLM models. For each event and task $k\in\mathcal K$, we assume that using an LLM can improve the efficiency of an operator with role $r\in\mathcal R$. When an operator with role $r$ uses model $m$ for task $k$ on an event of criticality $a$, the model returns an answer in a fraction $\eta_{k,r,a}\in[0,1]$ of the time a human would otherwise need. We call $\eta_{k,r,a}$ the \emph{speed ratio}; it includes the time the operator spends prompting the model and checking its output. The model's answer is correct with probability $E^m\in[0,1]$, which we call its \emph{accuracy}. With probability $1-E^m$, the answer is incorrect and the operator must redo the entire task manually. Therefore, the personnel time spent per unit of time by operators with role $r$ when using model $m\in\mathcal M$ is

\begin{equation}
\begin{split}
\tau_r^m = N\cdot\sum_{a\in\mathcal A} h_a \cdot p_a \cdot \sum_{k\in\mathcal K} \Big[ & E^m\cdot\eta_{k,r,a}\cdot\tau_{k,r,a} \\
&+ (1-E^m)\cdot(1+\beta_{k,r,a})\cdot\tau_{k,r,a}\Big]
\end{split}
\end{equation}
where $\beta_{k,r,a}\geq0$ is the additional rework penalty for an incorrect answer. We assume that incorrect answers are detected and redone, so that final output quality remains at the human baseline~(Sec.~\ref{sec:methodological}). The overall cost per unit of time is then
\begin{equation}
\label{eq:Cm}
C^m = \underbrace{\sum_{r\in\mathcal R} c_r\,\tau_r^m}_{C^m_\textit{pers}} + C^m_\text{tech} + C_\text{inv}
\end{equation}
where $C_\textit{pers}^m$ denotes the operational personnel cost per unit of time when using model $m$, and $C^m_\text{tech}$ is the technology cost, detailed below. The formulation also accommodates multi-model deployments by assigning a model $m_{k,a,r}$ to each task, criticality level, and role. The single-model case is recovered when this assignment is constant. \\

\subsubsection{Technology cost ($C^m_\text{tech}$\,)} The technology cost comprises inference and infrastructure costs:
\begin{align}
C^m_\text{tech} = C^m_\text{tok} + C^m_\text{infra}
\end{align}

The formulation covers both cloud and self-hosted deployments. For a cloud API, inference is billed per million tokens and no infrastructure is owned ($C^m_{\text{infra}} = 0$). For a self-hosted deployment, we assume that GPUs are rented by the hour from a cloud provider. Token pricing is then zero ($C^m_{\text{tok}} = 0$), and the cost comes entirely from renting the GPUs. This makes self-hosting, like API use, a recurring operational cost. One-off expenditures, such as purchasing on-premises hardware, remain part of $C_\text{cap}$ (Sec. \ref{sec:operational-savings-and-return-of-investment}). Hybrid configurations are also covered by combining the two terms. \\

The \textbf{inference cost} $C^m_\text{tok}$ sums the per-event token costs over all events:
\begin{align}
C^m_\text{tok} = N \sum_{a\in\mathcal A} h_a\, p_a \sum_{k\in\mathcal K}\sum_{r\in\mathcal R} c^m_{k,r,a}
\end{align}
The \textit{per-event cost}~$c^m_{k,r,a}$ is defined as follows. An operator with role $r$ sends $\textit{req}_{k,r,a}$ requests to the model for task~$k$ on an event of criticality~$a$. Each request uses $\textit{tok}_{k,r,a}^\text{in}$ input tokens and $\textit{tok}_{k,r,a}^\text{out}$ output tokens, priced at $\gamma^{m,\text{in}}$ and $\gamma^{m,\text{out}}$ per million tokens, respectively. Thus, the per-event cost is
\begin{align}
c_{k,r,a}^m = 10^{-6}\cdot\textit{req}_{k,r,a}\cdot(\textit{tok}_{k,r,a}^\text{in}\cdot\gamma^{m,\text{in}} + \textit{tok}_{k,r,a}^\text{out}\cdot\gamma^{m,\text{out}})
\end{align} 

The \textbf{infrastructure cost} $C^m_{\text{infra}}$ over the analysis window $\Delta T$ (in hours) is
\begin{align}
C^m_{\text{infra}} = n^m_\text{gpu} \cdot c^m_\text{gpu} \cdot \Delta T
\end{align}
where $c^m_{\text{gpu}}$ is the amortized hourly cost per GPU for model $m$, and $n^m_{gpu}$ is the number of GPUs required to run it:
\begin{align}
\label{eq:ngpu}
n^m_{gpu} = \underbrace{\left\lceil \frac{2\delta_\text{tot}^m}{V^m_\text{gpu}} \right\rceil}_{n^m_\text{mem}} \cdot \underbrace{\left\lceil \frac{D^m}{\phi^m \cdot \Delta T \cdot 3600} \right\rceil}_{n^m_\text{rep}}
\end{align}

$n^m_\text{mem}$ is the \emph{minimum deployment unit}, or the smallest number of GPUs that can hold one \emph{replica} of model $m$ (one copy of the model loaded in memory and ready to serve requests). $\delta^m_\text{tot}$ is the \emph{total} number of model parameters. The factor of 2 reflects FP16/BF16 precision, which uses 2 bytes per parameter. $V^m_\textit{gpu}$ is the memory capacity, in bytes, of the GPU used for model $m$.

$n^m_\text{rep}$ is the \textit{replication factor}, or the number of model replicas required to sustain the total token demand $D^m$ per unit of time. This demand is defined as:
\begin{align}
D^m = N \sum_{a\in\mathcal A} h_a\, p_a \sum_{k\in\mathcal K} \sum_{r\in\mathcal R} \textit{req}_{k,r,a}\,(\textit{tok}_{k,r,a}^\text{in} + \textit{tok}_{k,r,a}^\text{out})
\end{align}
$\phi^m$ is the number of tokens per second (tok/s) that one replica can serve. Compute per token scales with the number of \emph{active} parameters, $\delta^m_{\text{act}}$, which equals $\delta^m_{\text{tot}}$ for dense models. We therefore approximate $\phi^m$ as inversely proportional to the active parameter count and directly proportional to the number of GPUs in a replica, $n^m_\text{mem}$. The approximation is anchored to a reference model with $\delta_\text{ref}$ parameters:
\begin{align}\label{eq:phi}
\phi^m = \phi_{\text{ref}}\cdot \frac{\delta_{\text{ref}}}{\delta^m_{\text{act}}}\cdot n^m_\text{mem}
\end{align}
Using one measured reference places all models on a common basis. However, this assumes that the reference model's efficiency transfers across model sizes and scales linearly with the number of GPUs in a replica, even though parallel efficiency is sublinear.

The model does not account for serving optimizations such as request batching or key-value caching. It therefore overestimates GPU requirements and provides an upper bound on cost.

\subsection{Operational Savings, Capital Expenditure Ratio, and Return on Investment (ROI)}
\label{sec:operational-savings-and-return-of-investment}
The \textit{operational cost savings} per unit of time from using LLM~$m$ are defined as:
 \begin{align}
 \label{eq:cost-savings}
 \Delta C^m \coloneqq C\nollm-C^m,
 \end{align}
where the terms are defined in~\eqref{eq:Cnollm} and \eqref{eq:Cm}.

We denote by $C_\text{cap}$ the per-unit-of-time equivalent of the \textit{additional capital cost}. It captures one-off adoption costs, including changes to workflows, procedures, and roles; software and infrastructure setup; and employee training. This cost is difficult to estimate, varies substantially across companies, and depends on the LLM framework adopted. Its precise assessment is outside our scope, so we treat it as a parameter in amortized form. For simplicity, we amortize these expenditures over three years, consistent with standard practice in IT and digital infrastructure economics \cite{EBA2020SoftwareRTS,IFRS2021IAS38} and reflecting rapid technological obsolescence.

We define the \emph{additional capital expenditure ratio} as:
\begin{align}
\label{eq:theta}
\theta \coloneqq C_\text{cap}/C\nollm.
\end{align}

The \textit{Return on Investment (ROI)} is defined as follows~\cite{UNESCO2020_ROI}:
\begin{align}
\label{eq:ROI}
\textit{ROI}^m \coloneqq \frac{\Delta C^m-C_\text{cap}}{C_\text{cap}} 
\end{align}

The following proposition gives an upper threshold on the capital expenditure ratio $\theta$ compatible with a target return $y$ on investment in LLM adoption.\footnote{Proof in Appendix~\ref{app:ROI-proof}.}

\begin{proposition}
\label{prop:ROI}
For any target return $y$, the ROI exceeds $y$ if and only if the capital expenditure ratio $\theta$ stays below a threshold:
\begin{align}
\textit{ROI}^m>y \iff \theta < \frac{\Delta C^m}{(1+y) \cdot C\nollm}
\end{align}
\end{proposition}
The bound expresses the adoption decision as a threshold on $\theta$, which is known before adoption.

The model above does not account for the cost of incorrect answers that escape detection. If this occurs at an expected rate $\lambda$ per unit of time, each costing $C_\text{miss}$ on average, the expected savings become $\Delta C^m-\lambda \cdot C_\text{miss}$, and Proposition~\ref{prop:ROI} with $y=0$ bounds the rate that adoption can absorb:
\begin{align}
\label{eq:breach-budget}
\textit{ROI}^m>0 \iff \lambda < \lambda_{\max} = \frac{\Delta C^m-\theta\,C\nollm}{C_\text{miss}}
\end{align}

\subsection{Uncertainty Characterization}
Evaluating the economic impact of a technology that is not yet widely deployed in the organization under study is challenging because of \emph{uncertainty}. This uncertainty reflects both limited knowledge about the appropriate values of certain quantities and their day-to-day variability. We therefore use Probabilistic Sensitivity Analysis (PSA) via Monte Carlo simulation, a standard method in economic assessment and decision-making \cite{ProbabilisticSensitivityAnalysis2021}. All probabilities and risk measures reported below are estimated from $10^4$ independent Monte Carlo draws for each configuration.

We model the uncertain parameters as random variables using two types of distributions~\cite{Chadly2022UncertaintySimulationCostAnalyses, GeorgeWilliams2022ProbabilisticSmartHubs}:
\begin{itemize}
    \item \emph{Beta} distributions in $[0,1]$ model the speed ratio $\eta_{k,r,a}$ and accuracy $E^m$, as is standard for quantities bounded in this range.\footnote{A Beta's CV is capped at the maximum feasible for its mean.}
    \item \emph{Gamma} distributions in $[0,\infty)$ model operator effort $\tau_{k,a}$, event volume $N$, input and output token counts $\textit{tok}^{\text{in}}_{k,r,a}$ and $\textit{tok}^{\text{out}}_{k,r,a}$, request counts $\textit{req}_{k,r,a}$, and the rework penalty $\beta_{k,r,a}$. These distributions allow a non-zero probability of large values, such as very high token counts.
\end{itemize}

The remaining parameters are deterministic. The severity mix $p_a$ and true-positive rate $h_a$ describe the event stream. The organization sets the role allocation $\alpha_{a,r}$, per-role cost $c_r$, and analysis window $\Delta T$. The capital expenditure ratio $\theta$ is the decision variable and is known at adoption time. Technical parameters, including token prices $\gamma^{m,\text{in}}$ and $\gamma^{m,\text{out}}$, model sizes $\delta^m_\text{tot}$ and $\delta^m_\text{act}$, GPU memory $V^m_\text{gpu}$, throughput $\phi^m$, and GPU cost $c^m_\text{gpu}$, are fixed specifications selected during the design phase. The non-operational cost $C_\text{inv}$ cancels in $\Delta C^m$; we set it to zero in the instantiation.

We denote the means of these random variables by $\bar\eta_{k,r,a}$, $\bar\beta_{k,r,a}$, $\bar E^m$, $\bar\tau_{k,a}$, $\bar N$, $\bar{\textit{tok}}^{\text{in}}_{k,r,a}$, $\bar{\textit{tok}}^{\text{out}}_{k,r,a}$, and $\bar{\textit{req}}_{k,r,a}$. Their coefficients of variation (CVs), denoted by $\textit{CV}(\cdot)$, are the standard deviations divided by the means and represent the relative uncertainty of each parameter.\footnote{The conversion from means and CVs to Beta and Gamma sampling parameters is given in App.~\ref{app:par-proof}.}

We use a common CV for all uncertain parameters. Without empirical evidence about their relative uncertainty, assigning different CVs would impose an unsupported structure and suggest unwarranted precision. Instead, we treat the CV as a control variable and examine a range of values:
\begin{align}
\label{eq:uncertainty-levels}
&\text{low }(\textit{CV}{=}5\%),\quad \text{moderate }(\textit{CV}{=}15\%),\nonumber\\
&\text{high }(\textit{CV}{=}25\%),\quad \text{extreme }(\textit{CV}{=}35\%).
\end{align}
At higher values, the analysis is not informative because the output distributions overlap. This indicates that, beyond the extreme level, the economic benefits of LLM adoption cannot be reliably assessed.

The same principle applies to the means. For each parameter indexed by task, role, or severity, we use a common mean across indices unless evidence supports differentiating them, for example by severity. In each realization, we still sample each indexed value independently. Thus, common means avoid unsupported bias while independent draws capture uncertainty. The specific values are given in Sec.~\ref{sec:par_calibration} and summarized in Tab.~\ref{tab:input}.

\subsection{Profitability and Value at Risk}
We assess how confidently the benefits of LLM adoption can be established despite uncertainty about deployment and impact. We quantify this confidence using the \emph{probability of profitability}, defined as the probability of a positive return at a given uncertainty level:
\begin{align}
p_\textit{CV}(\theta) &\coloneqq \mathbb P\!\left(\textit{ROI}^m>0 \mid \textit{CV}\right) \nonumber\\
&= \mathbb P\!\left(\theta < \Delta C^m / C\nollm \mid \textit{CV}\right)
\end{align}
where the second equality follows from Proposition~\ref{prop:ROI} with $y{=}0$. This defines a decision rule for a risk-averse organization, which adopts LLMs only if the probability of profitability exceeds a required confidence level. At each uncertainty level~\eqref{eq:uncertainty-levels}, $p_\textit{CV}$ depends on the capital expenditure ratio $\theta$~\eqref{eq:theta}; adoption is warranted when $\theta$ falls below a corresponding threshold. Because $\theta$ is known before adoption, this decision can be made deterministically. Thus, epistemic uncertainty about LLM performance does not prevent a deterministic go/no-go decision.

We treat LLM performance as independent of $\theta$. A larger capital investment might improve performance; for example, better employee training could lower the speed ratio $\bar\eta_{k,r,a}$ or increase accuracy $\bar E^m$ through more effective prompting. We do not assume such an improvement, because no evidence currently quantifies the relationship between capital expenditure and LLM performance. The closest study we found examines training quality in a different domain and reports mixed effects~\cite{Feuerriegel2026}.

We quantify downside risk using the \textit{Value at Risk at level $z$} ($\text{VaR}_z$)~\cite{McNeilFreyEmbrechts2015}:
\begin{align}
\text{VaR}_z(\textit{ROI}) = -\inf\left\lbrace y\in\mathbb R : \mathbb P(\textit{ROI}<y) > z \right\rbrace
\end{align}
For example, $\text{VaR}_{10\%}(\textit{ROI})=20\%$ means there is a $10\%$ probability that LLM adoption will incur a loss of at least $20\%$ ($\textit{ROI}\le -20\%$) relative to non-adoption.

\section{Parameter Calibration}\label{sec:par_calibration}

\begin{table}
\centering
\footnotesize
\setlength{\tabcolsep}{3pt}
\caption{Model input parameters and their simulation setting. Distr.: B = Beta, G = Gamma, D = Deterministic.}
\label{tab:input}
\resizebox{\columnwidth}{!}{%
\begin{tabular}{@{}llllccc@{}}
\toprule
\textbf{Symbol} & \textbf{Name} & \textbf{Unit} & \textbf{Constraint} & \textbf{Distr.} & \textbf{Sec.} & \textbf{Tab.} \\
\midrule
\multicolumn{7}{@{}l}{\textbf{Model Instantiation}} \\
$\mathcal R$ & Roles & --- & $|\mathcal R|{=}3$ & D & \ref{sec:model_sets} & --- \\
$\mathcal A$ & Severity levels & --- & $|\mathcal A|{=}3$ & D & \ref{sec:model_sets} & --- \\
$\mathcal K$ & Tasks & --- & $|\mathcal K|{=}3$ & D & \ref{sec:model_sets} & --- \\
$\mathcal M$ & LLM archetypes & --- & $|\mathcal M|{=}6$ & D & \ref{sec:model_sets} & --- \\
\midrule
\multicolumn{7}{@{}l}{\textbf{Personnel Cost}} \\
$c_r$ & Personnel cost & \$/h & --- & D & \ref{sec:personnel_cost} & \ref{tab:soc_salary_detailed} \\
\midrule
\multicolumn{7}{@{}l}{\textbf{Alert Volume and Severity}} \\
$N$ & Alert volume & 1/day & --- & G & \ref{sec:alert_severity} & \ref{tab:alert-distribution} \\
$p_a$ & Severity share & --- & $\sum_a p_a{=}1$ & D & \ref{sec:alert_severity} & \ref{tab:alert-distribution} \\
$h_a$ & TP rate & --- & --- & D & \ref{sec:alert_severity} & \ref{tab:alert-distribution} \\
$\Delta T$ & Analysis window & h & $=24$ & D & \ref{sec:alert_severity} & --- \\
\midrule
\multicolumn{7}{@{}l}{\textbf{Analyst Effort and Allocation}} \\
$\tau_{k,a}$ & Effort & min & --- & G & \ref{sec:task_effort} & \ref{tab:effort-alloc} \\
$\alpha_{a,r}$ & Role allocation & --- & $\sum_r \alpha_{a,r}{=}1$ & D & \ref{sec:task_effort} & \ref{tab:effort-alloc} \\
\midrule
\multicolumn{7}{@{}l}{\textbf{LLM Efficiency and Rework}} \\
$\eta_{k,r,a}$ & Speed ratio & --- & $=\eta$ & B & \ref{sec:efficiency} & --- \\
$\beta_{k,r,a}$ & Rework penalty & --- & $=\beta$ & G & \ref{sec:efficiency} & --- \\
\midrule
\multicolumn{7}{@{}l}{\textbf{LLM Model Archetypes}} \\
$E^m$ & Accuracy & --- & --- & B & \ref{sec:archetypes} & \ref{tab:archetypes} \\
$\gamma^{m,\text{in/out}}$ & Token price & \$/Mtok & --- & D & \ref{sec:archetypes} & \ref{tab:archetypes}  \\
$\delta^m_\text{tot}$ & Tot size & --- & --- & D & \ref{sec:archetypes} & \ref{tab:archetypes}  \\
$\delta^m_\text{act}$ & Active size & --- & --- & D & \ref{sec:archetypes} & \ref{tab:archetypes}  \\
\midrule
\multicolumn{7}{@{}l}{\textbf{Infrastructure and Inference Demand}} \\
$c^m_\text{gpu}$ & GPU cost & \$/h & $=c_{gpu}$ & D & \ref{sec:infra_infer} & --- \\
$V^m_\text{gpu}$ & GPU memory & bytes & $=V_{gpu}$ & D & \ref{sec:infra_infer} & --- \\
$\phi^m$ & Throughput & tok/s & Eq.~\eqref{eq:phi} & D & \ref{sec:infra_infer} & --- \\
$\textit{tok}^{\text{in/out}}_{k,r,a}$ & Tokens & tok & $=\textit{tok}_a$ & G & \ref{sec:infra_infer} & \ref{tab:tokens} \\
$\textit{req}_{k,r,a}$ & Requests & --- & $=\textit{req}_a$ & G & \ref{sec:infra_infer} & \ref{tab:tokens} \\
\midrule
\multicolumn{7}{@{}l}{\textbf{Economic Parameters}} \\
$\theta$ & Capex ratio & --- & --- & D & \ref{sec:theta_calibration} & --- \\
$C_\text{inv}$ & Non-op.\ cost & \$/day & $=0$ & D & \ref{sec:theta_calibration} & --- \\
$\lambda$         & Miss rate & 1/day & $=0$               & D & \ref{sec:theta_calibration} & --- \\
\bottomrule
\end{tabular}}
\end{table}

This section reports the parameter values used in the evaluation (Tab. \ref{tab:input}). From this point onward, we use SOC terminology: events are \emph{alerts}, and their criticality is \emph{severity}. Where authoritative market benchmarks were available, we used them to set parameter values. The remaining values are based on expert judgment and are modeled as uncertain in the Monte Carlo simulation. For stochastic parameters, the value reported below is the mean of the corresponding distribution.

\subsection{Model Instantiation ($\mathcal R$, $\mathcal A$, $\mathcal K$, $\mathcal M$)} \label{sec:model_sets}
We instantiate the abstract sets in Sec.~\ref{sec:general-cost-model} using the entities in the reference SOC structure (Sec.~\ref{sec:reference-soc}). The role set contains the three analyst levels, $\mathcal R=\{\text{L1}, \text{L2}, \text{L3}\}$. The severity set is $\mathcal A=\{\text{low}, \text{medium}, \text{high}\}$. Critical alerts are excluded because we assume they are redirected to a dedicated incident-response team outside the SOC. The task set covers the detection and analysis activities in scope: triage, analysis, and documentation, denoted by $\mathcal K=\{\text{tri}, \text{ana}, \text{doc}\}$. The model set $\mathcal M$ contains six deployment archetypes: three API-based (\textsc{API-High}, \textsc{API-Mid}, and \textsc{API-Low}) and three self-hosted (\textsc{SH-High}, \textsc{SH-Mid}, and \textsc{SH-Low}). Together, they span the cost--accuracy--deployment space. Their calibrated values are given in Sec.~\ref{sec:archetypes}. Whether the GPUs are rented or owned affects only whether their cost is included in $C^m_\text{infra}$ or $C_\text{cap}$ (Sec.~\ref{sec:general-cost-model}).

\subsection{Personnel Cost ($c_r$)} \label{sec:personnel_cost}

\begin{table}[pos=b]
\centering
\caption{Personnel Cost $c_r$ (\$/h) by SOC role. Source: BLS OEWS May 2025 (SOC 15-1212); $c_r$: mean of the band's FL percentiles.}
\label{tab:soc_salary_detailed}
\begin{tabular*}{\columnwidth}{@{\extracolsep{\fill}}lcccc@{}}
\toprule
 & \textbf{P10} & \textbf{P25} & \textbf{P50} & \textbf{P75} \\
\midrule
Base & 36.1 & 47.0 & 62.1 & 78.6 \\
Fully Loaded  & 51.6 & 67.3 & 88.8 & 112.4 \\
\midrule
$c_{L1}$ & \multicolumn{2}{c}{\cellcolor{gray!20}59.4} &  &  \\
\addlinespace[3pt]
$c_{L2}$ &  & \multicolumn{2}{c}{\cellcolor{gray!20}78.0} &  \\
\addlinespace[3pt]
$c_{L3}$ &  &  & \multicolumn{2}{c}{\cellcolor{gray!20}100.6} \\
\bottomrule
\end{tabular*}
\end{table}

Salary data for SOC analysts come from the U.S. Bureau of Labor Statistics (BLS) Occupational Employment and Wage Statistics (OEWS) survey for Information Security Analysts (SOC code 15-1212, May 2025) \cite{bls_oews_2025}. We use these data because they report occupation-specific wage percentiles, which are needed for the mapping to L1--L3 described below.

Fully Loaded (FL) costs apply a $1.43\times$ burden multiplier to base salary. This accounts for employer payroll taxes, health insurance, retirement contributions, and paid time off, and is consistent with BLS estimates that benefits make up roughly 30\% of total compensation in private industry~\cite{bls_ecec_2026}. The resulting values are also in line with industry salary surveys for security operations roles~\cite{isc2_sscp_2024,indeed_soc_2025,glassdoor_soc_2025}.

We map the BLS percentiles to SOC levels as follows: L1 corresponds to P10--P25, L2 to P25--P50, and L3 to P50--P75. For each level, we use the mean hourly wage as the base salary input. Its Fully Loaded value is the per-role cost $c_r$ used in the model (Tab. \ref{tab:soc_salary_detailed}).

Emerging AI/ML roles, such as AI Security Engineer, are included in implementation costs. For operational expenditures, however, we treat AI/ML competencies as cross-functional skills within existing SOC layers rather than as dedicated positions. This approach is consistent with the NICE Framework, which defines AI Security as a Competency Area (NF-COM-002), not a Work Role~\cite{nist_nice_2026}.

\subsection{Alert Volume and Severity ($N$, $p_a$, $h_a$, $\Delta T$)} \label{sec:alert_severity}

\begin{table}
\centering
\caption{Alert volume $N$, alert distribution $p_a$ and true-positive rate $h_a$. Baseline $N=1{,}000$ alerts/day, critical incidents not in scope.}
\label{tab:alert-distribution}
\begin{tabular*}{\columnwidth}{@{\extracolsep{\fill}}lrrrr@{}}
\toprule
\textbf{Severity} & \textbf{Proportion} & \textbf{Alerts} & \textbf{TP Rate} & \textbf{Incidents} \\
$a\in\mathcal{A}$ & $p_a$ & $p_a N$ & $h_a$ & $p_a h_a N$ \\
\midrule
Low    & 0.85 & 850   & 0.15 & 128 \\
Medium & 0.13 & 130   & 0.50 & 65  \\
High   & 0.02 & 20    & 0.75 & 15  \\
\midrule
\textbf{Total} & --- & 1{,}000 & --- & \textbf{208} \\
\bottomrule
\end{tabular*}
\end{table}

Daily alert volume varies by orders of magnitude across SOCs, and no standardized benchmark exists. Reported volumes range from a few hundred to well over ten thousand alerts per day, depending on organization size and detection tuning~\cite{paloalto2020, vectra2023, thehackernews2025}. The median organization receives roughly 100 alerts per day, while the largest environments approach 1{,}000~\cite{prophet2025}. We use $N=1{,}000$ alerts/day as our baseline, representing large enterprises that are among the most likely candidates for AI adoption. Figure~\ref{fig:roivsN} and the sensitivity analysis in Sec.~\ref{sec:res_sensitivity} show how other alert volumes affect the results.
We set the analysis window to $\Delta T = 24$\,h (one day). All per-unit-of-time quantities refer to this interval.

We model three severity levels: low, medium, and high. The distribution $p_a$ is based on the authors' operational experience in security operations and was reviewed by senior SOC practitioners for representativeness. Critical incidents are excluded from the modeled workload because handling them varies substantially; we assume they are delegated to an external incident-response team (Sec.~\ref{sec:reference-soc}). At each severity level, we model only the true-positive fraction $h_a$ that reaches human analysts. Upstream pre-triage automation, including SOAR playbooks, auto-enrichment, and reputation lookups, handles raw event qualification. Thus, triage in $\mathcal K$ means that analysts validate and confirm the severity of alerts that reach them. The simulation values are reported in Tab.~\ref{tab:alert-distribution}.

\subsection{Analyst Effort and Allocation ($\tau_{k,a}$, $\alpha_{a,r}$)}\label{sec:task_effort}

Effort by analyst task is difficult to benchmark, and no standardized measurement method exists (Sec.~\ref{sec:reference-soc}). Available industry estimates vary widely across SOCs and mostly measure \emph{elapsed} time (MTTP, MTTA, MTTR)~\cite{vectra2023, ponemon_secops2026, prophet2025}, rather than the \emph{work effort} required by our model. They also provide no task-level breakdown. We therefore use effort values based on operational experience, reviewed by senior SOC practitioners for representativeness, and treat them as uncertain in the Monte Carlo analysis. We assume that each alert's total effort ($\tau_a$) is divided among tasks as follows: 20\% for triage, 50\% for analysis, and 30\% for documentation. Analysis receives the largest share because it involves the investigative work. We also allocate effort across roles according to severity. L1 analysts handle most low-severity alerts, L1 and L2 share medium-severity alerts, and L2 and L3 carry most of the high-severity effort. This allocation reflects the escalation of more complex work to more senior analysts (Sec.~\ref{sec:reference-soc}). Tab.~\ref{tab:effort-alloc} reports the resulting per-task effort $\tau_{k,a}$ and role allocations $\alpha_{a,r}$.

\begin{table}
\centering
\caption{Per-alert effort $\tau_{k,a}$ (min, converted to hours in the model) and role allocation $\alpha_{a,r}$ by severity. $\tau_a=\sum_k\tau_{k,a}$ and $\sum_r\alpha_{a,r}=1$.}
\label{tab:effort-alloc}
\resizebox{\columnwidth}{!}{%
\begin{tabular}{lrrrrrrr}
\toprule
\multirow{2}{*}{\textbf{Severity}} & \multicolumn{4}{c}{\textbf{Effort $\tau_{k,a}$ (min)}} & \multicolumn{3}{c}{\textbf{Allocation $\alpha_{a,r}$}} \\
\cmidrule(lr){2-5} \cmidrule(lr){6-8}
 & $\tau_a$ & $\tau_{\text{tri},a}$ & $\tau_{\text{ana},a}$ & $\tau_{\text{doc},a}$ & $\alpha_{a,L1}$ & $\alpha_{a,L2}$ & $\alpha_{a,L3}$ \\
\midrule
Low    & 10 & 2  & 5  & 3  & 0.8 & 0.2 & 0.0 \\
Medium & 30 & 6  & 15 & 9  & 0.5 & 0.5 & 0.0 \\
High   & 60 & 12 & 30 & 18 & 0.1 & 0.6 & 0.3 \\
\bottomrule
\end{tabular}%
}
\end{table}

\subsection{LLM Efficiency and Rework ($\eta$, $\beta$)}\label{sec:efficiency}

No standardized benchmark measures LLM efficiency in security operations. Reported gains vary widely with the task, user expertise, and measurement setup. Available figures mostly measure elapsed time rather than analyst work effort, provide no task- or role-level breakdown, and depend heavily on the context in which they were measured.

Across knowledge work, controlled studies report time savings ranging from about 14\% in customer support~\cite{brynjolfsson_generative_2023} to 26--56\% in software development~\cite{cui_effects_2025, peng_impact_2023}. In security operations, a ten-month field study of 45 analysts found that LLMs were adopted as sensemaking aids and reduced workload~\cite{singh_llms_2025}. A randomized trial found that analysts worked 22--26\% faster~\cite{edelman_randomized_2024}, while live-operations telemetry associated LLM assistance with an approximately 30\% reduction in Mean Time to Resolve~\cite{bono_generative_2024}. We use these findings only to bound a single aggregate speed ratio.

We also assume that LLM adoption does not change the SOC's task or role structure. The task set $\mathcal K$ and role set $\mathcal R$ are the same with and without LLMs ($\mathcal K^{noLLM}{=}\mathcal K^m$, $\mathcal R^{noLLM}{=}\mathcal R^m$). Any additional operational activities introduced by LLMs, such as prompting and output validation, are incorporated into the coefficients in this section. These coefficients therefore represent performance net of that overhead.

Accordingly, we adopt a single, aggregate \textit{speed ratio} $\eta$, common to all tasks, roles, and severities:
\begin{equation}
\eta_{k,r,a} = \eta, \qquad \forall\, k\in\mathcal K,\ r\in\mathcal R,\ a\in\mathcal A .
\end{equation}
We write $\eta = 1 - \epsilon$, where $\epsilon$ is the fractional time reduction reported in the literature. Evidence from security operations places $\epsilon$ between $22$\% and $30$\%, corresponding to $\eta \in [0.70, 0.78]$. We use $\eta = 0.75$, equivalent to a $25\%$ reduction, which is approximately the midpoint of this range. We treat $\eta$ as uncertain in the Monte Carlo analysis.

The evidence above does not isolate the effort required to correct incorrect LLM output. Our model captures this effort through the rework penalty $\beta$. The cost may be more than additive: LLMs can omit critical details or introduce factual errors, and participants found inaccurate drafts harder to edit than writing from scratch~\cite{kramer2025integrating}. A draft can therefore have a negative net contribution. Our model captures this discontinuity through the interaction between accuracy ($E^m$) and the rework penalty ($\beta$). The downside can be substantial. A controlled trial found that experienced developers were $19\%$ slower with AI assistance~\cite{becker_measuring_2025}, and a security-operations trial found that experts were about $26\%$ slower on a response task~\cite{edelman_randomized_2024}. Errors and over-reliance contribute to this effect; users may accept incorrect output at moderate levels of familiarity with AI~\cite{horowitz_bending_2024}.

Because no benchmark is available, we use a common \textit{rework penalty} $\beta$ across tasks, roles, and severities:
\begin{equation}
\beta_{k,r,a} = \beta, \qquad \forall\, k\in\mathcal K,\ r\in\mathcal R,\ a\in\mathcal A
\end{equation}
As our reference, we use the conservative case $\beta=\eta=0.75$. In this case, an incorrect answer is detected only after the LLM-assisted task is complete and requires a full manual redo. The time spent on the assisted attempt is therefore added to the redo time.

We also examine two limiting regimes in the simulations. When $\beta=0$, incorrect outputs are immediately recognizable and do not increase redo time. When $\eta=0$, correct outputs require no residual analyst effort. In each regime, the parameter set to zero is deterministic.

\subsection{LLM Archetypes ($\mathcal M$, $E^m$, $\gamma^{m,\text{in/out}}$, $\delta^m_\text{tot}$, $\delta^m_\text{act}$)}\label{sec:archetypes}

\begin{table}
\centering
\scriptsize
\setlength{\tabcolsep}{3pt}
\caption{Selected benchmarks: domain, reference, and shared-model overlap (blue heatmap; diagonal = models per benchmark, off-diagonal = models in common).}
\label{tab:benchmarks}
\resizebox{\columnwidth}{!}{%
\begin{tabular}{@{}l l l l cccccc@{}}
\toprule
\multirow{2}{*}{\textbf{ID}} & \multirow{2}{*}{\textbf{Benchmark}} & \multirow{2}{*}{\textbf{Ref}} & \multirow{2}{*}{\textbf{Type}} & \multicolumn{6}{c}{\textbf{Shared models}} \\
\cmidrule(l){5-10}
 & & & & CSE & SB & CTI & CM & DF & CS \\
\midrule
CSE & CyberSOCEval & \cite{deason_cybersoceval_2025} & CTI/Malw. & \cellcolor{blue!32}7 & \cellcolor{blue!8}1 & 0 & \cellcolor{blue!8}1 & \cellcolor{blue!12}2 & \cellcolor{blue!8}1 \\
SB  & SecBench     & \cite{jing_secbench_2025}       & Know.     & \cellcolor{blue!8}1 & \cellcolor{blue!48}16 & \cellcolor{blue!16}3 & \cellcolor{blue!16}3 & \cellcolor{blue!20}4 & \cellcolor{blue!12}2 \\
CTI & CTIBench     & \cite{alam_ctibench_2024}       & CTI       & 0 & \cellcolor{blue!16}3 & \cellcolor{blue!25}5 & \cellcolor{blue!12}2 & 0 & \cellcolor{blue!8}1 \\
CM  & CyberMetric  & \cite{tihanyi_cybermetric_2024} & Know.     & \cellcolor{blue!8}1 & \cellcolor{blue!16}3 & \cellcolor{blue!12}2 & \cellcolor{blue!70}25 & \cellcolor{blue!8}1 & \cellcolor{blue!12}2 \\
DF  & DFIR-Metric  & \cite{cherif_dfir-metric_2025}  & DFIR/IR   & \cellcolor{blue!12}2 & \cellcolor{blue!20}4 & 0 & \cellcolor{blue!8}1 & \cellcolor{blue!42}14 & \cellcolor{blue!8}1 \\
CS  & CS-Eval      & \cite{yu_cs-eval_2025}          & Know.     & \cellcolor{blue!8}1 & \cellcolor{blue!12}2 & \cellcolor{blue!8}1 & \cellcolor{blue!12}2 & \cellcolor{blue!8}1 & \cellcolor{blue!60}20 \\
\bottomrule
\end{tabular}%
}
\end{table}

LLM performance evaluation is an active field that includes general benchmarks, leaderboards, and cost-versus-quality rankings~\cite{chang_survey_2024, artificialanalysis}. Our aim here is simply to identify plausible values for the simulations.

The number of domain-specific benchmarks for security operations is also growing. We retained benchmarks that meet four criteria: (i) they apply to SOC tasks, (ii) they report model accuracy, (iii) they rely primarily on multiple-choice (MCQ) questions, and (iv) they are documented benchmarks rather than vendor blog posts. Six benchmarks meet these criteria~\cite{deason_cybersoceval_2025, jing_secbench_2025, alam_ctibench_2024, tihanyi_cybermetric_2024, cherif_dfir-metric_2025, yu_cs-eval_2025}. They cover security knowledge, cyber-threat intelligence, and digital forensics and incident response (Tab.~\ref{tab:benchmarks}).

We combined the benchmark results into a list of 71 unique models (Appendix~\ref{appendix: model_params}). We standardized model names across benchmarks while keeping distinct variants separate, such as Llama-3-8B and Llama-3-70B, GPT-4o and GPT-4o-mini, and DeepSeek-V2 and DeepSeek-V3. For each model, we averaged the reported accuracy and recorded input and output token prices as of September 2026. For deprecated models with no published prices, we used the most recent prices archived by the Internet Archive.

These benchmarks are imperfect proxies. Many evaluated models have become obsolete since the benchmarks were released. The benchmarks also cover largely distinct model sets, with a mean pairwise Jaccard similarity of approximately $0.06$ \cite{jaccard1912}, computed from Tab.~\ref{tab:benchmarks}. In addition, they report different accuracy values for the same model and do not provide prices from the time of evaluation. Among API models, price and accuracy show no monotonic association (Spearman $\rho = 0.04$, close to zero~\cite{spearman1904}); several low-cost models match or outperform more expensive ones.

For these reasons, we do not rely on individual models. Instead, we represent the collected set with a small number of archetypes that span the cost--accuracy space and provide inputs for the simulation. For each deployment mode, we select representative models near the $90$th, $50$th, and $20$th accuracy percentiles. We then use each selected model's token prices, parameter counts, and size (Tab.~\ref{tab:archetypes}). The archetypes cover a range of accuracy levels, costs, and parameter counts, and include both dense and Mixture-of-Experts (MoE) models. They provide illustrative inputs, not a ranking of models, since effectiveness depends on the task (Sec.~\ref{sec:limitations}). Figure~\ref{fig:profit_reg} sweeps $\bar E^m$ across its full range, so the profitability condition can be evaluated at any accuracy level, regardless of the value assigned to a specific model.

\begin{table}
\centering
\footnotesize
\setlength{\tabcolsep}{3pt}
\caption{Illustrative LLM archetypes and their simulation parameters. $n^m_\text{gpu}$ is the number of H100s to serve each model per Eq.~\eqref{eq:ngpu}.}
\label{tab:archetypes}
\begin{tabular*}{\columnwidth}{@{\extracolsep{\fill}}l c c c c c c@{}}
\toprule
\textbf{Archetype} & $\boldsymbol{E^m}$ & $\boldsymbol{\gamma^{m,\text{in}}}$ & $\boldsymbol{\gamma^{m,\text{out}}}$ & $\boldsymbol{\delta^m_\text{tot}}$ & $\boldsymbol{\delta^m_\text{act}}$ & $\boldsymbol{n^m_\text{gpu}}$ \\
\midrule
\makecell[l]{API-High\\{\scriptsize\textcolor{black!55}{\textit{(Grok-3)}}}} & 0.92 & \$3 & \$15 & -- & -- & -- \\
\makecell[l]{API-Mid\\{\scriptsize\textcolor{black!55}{\textit{(gpt-4o-mini)}}}} & 0.87 & \$0.15 & \$0.60 & -- & -- & -- \\
\makecell[l]{API-Low\\{\scriptsize\textcolor{black!55}{\textit{(claude-3-7-sonnet)}}}} & 0.58 & \$3 & \$15 & -- & -- & -- \\
\makecell[l]{SH-High\\{\scriptsize\textcolor{black!55}{\textit{(Falcon-180B-Chat)}}}} & 0.88 & -- & -- & 180 & 180 & 5 \\
\makecell[l]{SH-Mid\\{\scriptsize\textcolor{black!55}{\textit{(Qwen-1.5-MoE-A2.7B-Chat)}}}} & 0.75 & -- & -- & 14 & 2.7 & 1 \\
\makecell[l]{SH-Low\\{\scriptsize\textcolor{black!55}{\textit{(llama-3-8b)}}}} & 0.45 & -- & -- & 8 & 8 & 1 \\
\bottomrule
\end{tabular*}
\end{table}
\subsection{Infrastructure and Inference Demand ($c_\text{gpu}$, $V_\text{gpu}$, $\phi^m$, $\mathit{tok}^{\mathrm{in/out}}_a$, $\mathit{req}_a$)}\label{sec:infra_infer}

This subsection calibrates technology costs, which combine the serving infrastructure required for self-hosted deployments with the model's inference workload.

We size the infrastructure using a single GPU type, the NVIDIA H100 ($V_\text{gpu}$ = 80~GB), and assume that its cost is the same across models ($c^m_\text{gpu}=c_\text{gpu}$). We set $c_{\text{gpu}} = \$8$/h, the average on-demand list price for a single 80 GB H100 across the three major cloud providers as of September 2026.\footnote{Reference prices: AWS p5.4xlarge at \$6.88/h~\cite{aws_p5_pricing_2026}, Azure NC40ads H100 v5 at \$6.98/h~\cite{azure_nc_h100_pricing_2026}, and Google Cloud a3-highgpu-1g at \$10.92/h~\cite{gcp_gpu_pricing_2026}.}
We anchor per-replica throughput to an 8B dense reference model, Llama-3.1-8B, and use a conservative rate of $\phi_{\text{ref}}=3{,}000$ tok/s at $\delta_{\text{ref}}=8$\,B.\footnote{This is consistent with MLPerf Inference~v5.1, which reports about $5{,}100$ output tok/s on a single H100 in FP8~\cite{mlperf_inference_2025}.} Applying this rate equally to input and output tokens overestimates the required number of GPUs, consistent with the upper-bound cost estimate in Sec.~\ref{sec:general-cost-model}. Tab.~\ref{tab:archetypes} reports the resulting $n^m_\text{gpu}$ for each model.

We assume that the number of requests per alert and the number of tokens per request depend on severity, but not on task or role. Thus, $\sum_{k,r}\mathit{req}_{k,r,a}=\mathit{req}_a$ and $\mathit{tok}^{\mathrm{in/out}}_{k,r,a}=\mathit{tok}^{\mathrm{in/out}}_a$. Both quantities increase with severity to reflect the deeper, multi-step investigation required for higher-impact incidents. We use a $4{:}1$ input-to-output ratio because contextual input, such as alert payloads and logs, dominates the concise generated output. These assumptions are based on operational experience and were reviewed by senior SOC practitioners. Tab.~\ref{tab:tokens} reports the simulation values.

\begin{table}
\centering
\footnotesize
\caption{Simulation values for the per-severity inference parameters ($\mathit{req}_a$, $\mathit{tok}^{\mathrm{in}}_a$, $\mathit{tok}^{\mathrm{out}}_a$), at a 4:1 input-to-output ratio.}
\label{tab:tokens}
\begin{tabular*}{\columnwidth}{@{\extracolsep{\fill}}lccc@{}}
\toprule
\textbf{Severity} $a$ & \textbf{Requests} $\mathit{req}_a$ & \textbf{Input} $\mathit{tok}^{\mathrm{in}}_a$ & \textbf{Output} $\mathit{tok}^{\mathrm{out}}_a$ \\
\midrule
Low    & 5  & 8K  & 2K \\
Medium & 10 & 12K & 3K \\
High   & 15 & 16K & 4K \\
\bottomrule
\end{tabular*}
\end{table}

\subsection{Economic Parameters ($\theta$, $C_{inv}$,  $\lambda$)}\label{sec:theta_calibration}
The economic parameters are deterministic. The decision-maker sets the capital expenditure ratio $\theta$, which is known at the time of adoption. Rather than choosing a single value, we sweep $\theta$ across its admissible range and report profitability and VaR as functions of it. We set the non-operational cost $C_{\mathit{inv}}=0$ because common cost components cancel in the cost differential. Under this assumption, $\theta$ expresses capital expenditure relative to the personnel cost affected by the LLM. We likewise set the baseline miss rate to zero ($\lambda=0$), since the miss term enters only as an additive offset on the savings.

\section{Results} 
We exercise the model on the SOC instantiation of Sec.~\ref{sec:par_calibration}, spanning realistic operating points with the archetypes of Tab.~\ref{tab:archetypes}. We show how the model can turn cost, accuracy, and uncertainty into an actionable adoption decision. Reporting the ROI as a full distribution over the capital expenditure ratio $\theta$, we locate the boundary between value-creating and loss-making configurations together with the loss risked on the wrong side, and the levers that move a configuration across it.

\subsection{Profitability under uncertainty}\label{sec:res_profit}

Figure~\ref{fig:proi} reports $P(\mathrm{ROI}>0)$ over the capex ratio $\theta$ as CV varies; the shaded band ($P(ROI>0)\ge 0.8$) marks high-confidence profitability. Every curve is monotone decreasing in $\theta$. Increasing CV broadens the transition, from a sharp threshold at CV$=5\%$ to a gradual decline at CV$=35\%$, and reduces the dispersion across archetypes.

\begin{figure*}
    \centering
    \includegraphics[width=1\linewidth]{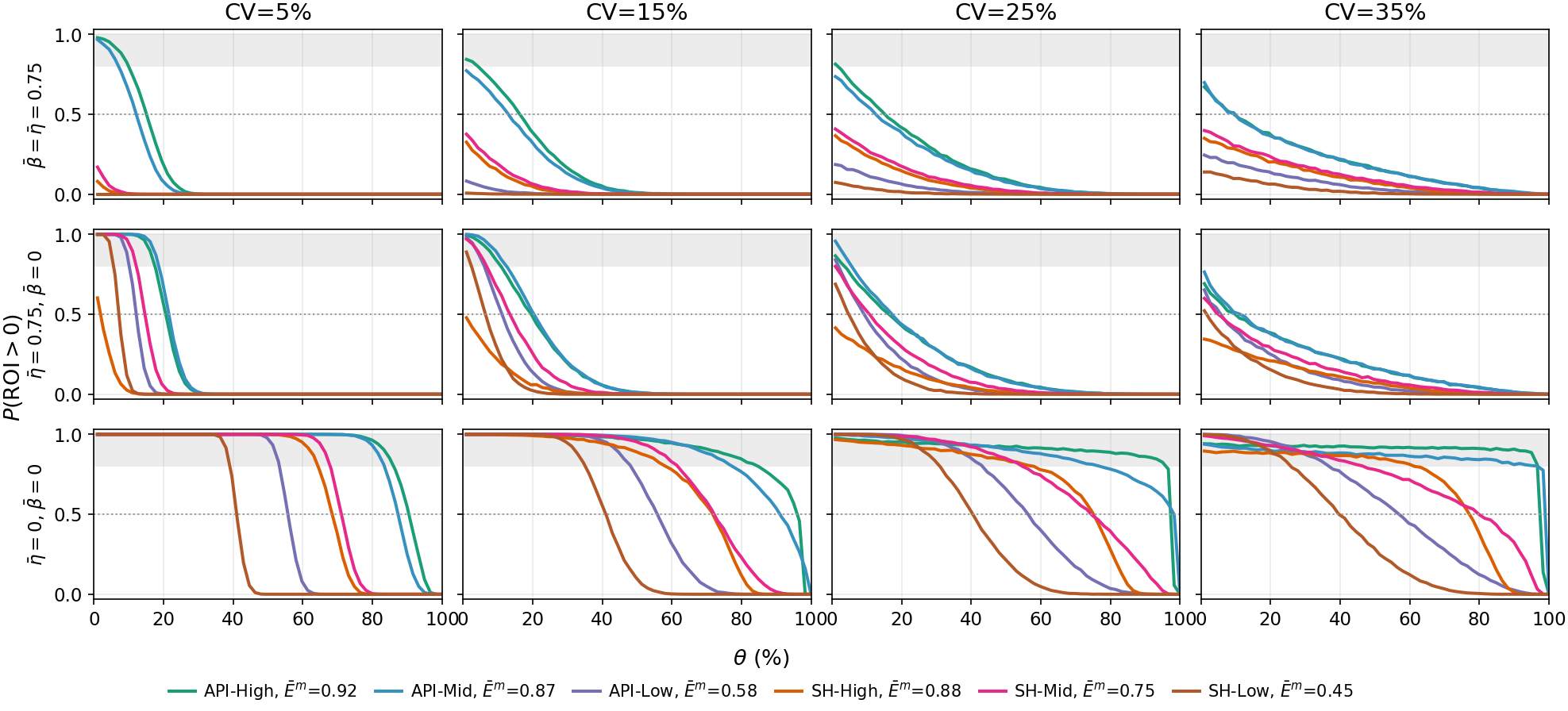}
    \caption{Probability that adoption is profitable $P(\mathrm{ROI}>0)$, as $\theta$ and CV vary. Rows: scenario $(\bar\eta,\bar\beta)$. Shaded band: $P\ge 0.8$.}
    \label{fig:proi}
\end{figure*}

\begin{figure}
    \centering
    \includegraphics[width=1\linewidth]{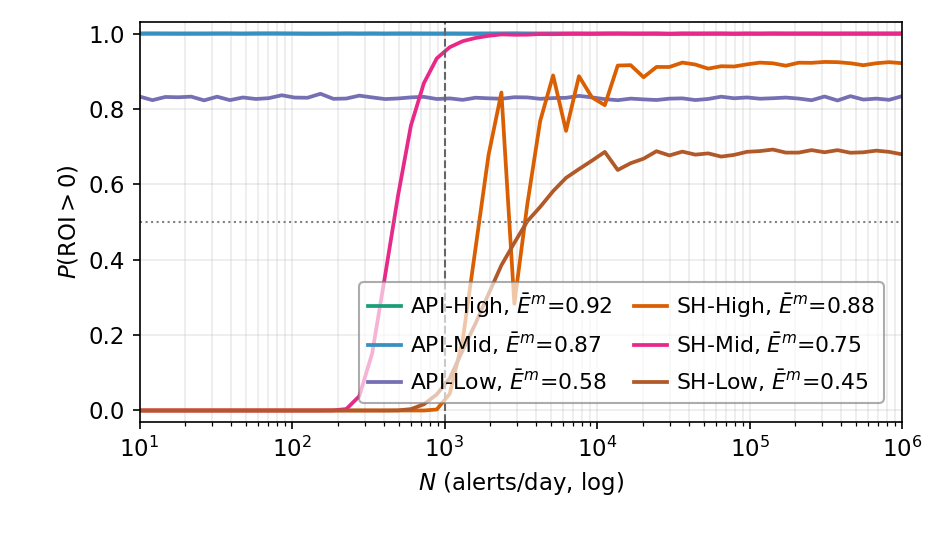}
    \caption{$P(\mathrm{ROI}>0)$ as the daily alert volume $N$ (log scale) varies, per archetype. $\bar\theta{=}10\%$,  $\bar\eta{=}0.75$, $\bar\beta{=}0$, CV${=}5\%$. Dashed line: baseline $\bar N{=}1000$.}
    \label{fig:roivsN}
\end{figure}

In the conservative case $\bar\beta=\bar\eta$ (top row: a wrong answer is caught late, costing the wasted attempt plus a full redo), only high-accuracy APIs enter the band at the lowest capex and uncertainty ($\theta<10\%$ at CV$=5\%$). With non negligible uncertainty (CV $\ge25\%$), no models guarantee profitability with high confidence.

In the $\bar\beta=0$ regime (mid row: a wrong output's impact is contained to its manual redo), the confident range extends to $\theta<20\%$ (CV$=5\%$), down to $10\%$ (CV$=25\%$), with API-Low and SH-Mid also entering the band at low capex. 

Under full automation $\bar\eta=0$ (bottom row: a cognitive-SOC regime with autonomy and human escalation only on critical cases) the confident range jumps to $\theta\approx 80$--$90\%$ for the leading models, and every archetype tolerates a far larger capex than in the previous scenarios. This scenario is also the most robust to uncertainty, holding high $P$ at large $\theta$ even at high CV and for weaker models. 

\begin{figure}
    \centering
    \includegraphics[width=1\linewidth]{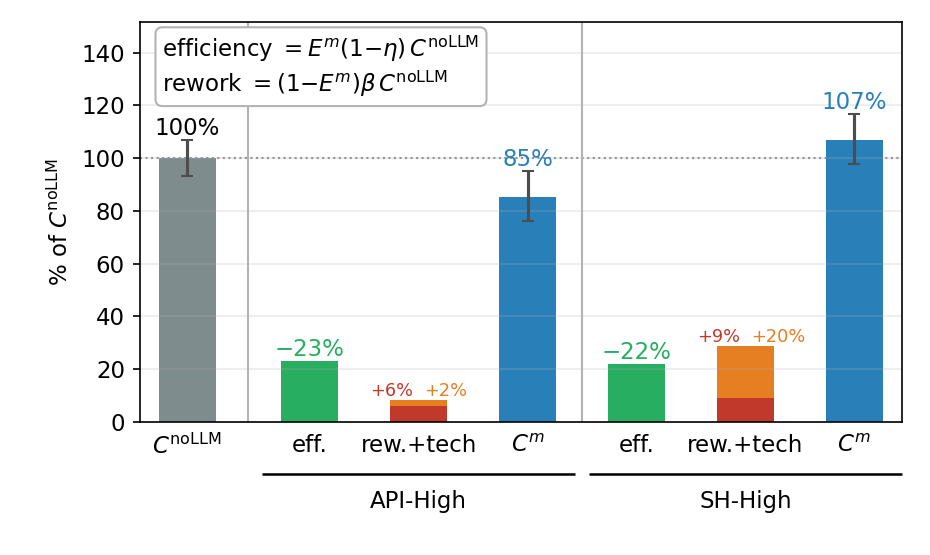}
    \caption{Example cost with LLM $C^m$ as \% of the no-LLM baseline $C^\mathrm{noLLM}$. $\bar\theta{=}10\%$, $\bar\beta{=}\bar\eta$, CV${=}5\%$, $\bar N{=}1000$. Whiskers: 10--90 percentile.}
    \label{fig:cost}
\end{figure}

Volume N impacts the decision only for self-hosted deployments (Fig.~\ref{fig:roivsN}): API $P(\mathrm{ROI}>0)$ is flat in $N$, since usage is billed per token and the return depends on accuracy and price. Self-hosted archetypes are a separate case: they carry a fixed GPU capex that is hard to amortize at such low alert volumes ($\bar N=1000$), with SH-High turning competitive at this volume only under full autonomy ($\bar\eta=0$). Self-hosted $P(\mathrm{ROI}>0)$ grows with $N$ in a sawtooth: each time the volume exceeds a replica's throughput an additional GPU replica is provisioned, a discrete capex step that depresses it before it is absorbed, most visibly where the per-replica GPU block is large. At high volume the plateau height is also affected by the GPU count, not accuracy alone, with SH-High, the most accurate self-hosted model, settling lower than SH-Mid because it runs on more GPUs.

A cost breakdown for the most accurate API and self-hosted archetypes is shown in Fig.~\ref{fig:cost}. Despite comparable efficiency savings, the two differ in the cost the LLM adds. For API-High this addition is negligible and comes mostly from rework rather than from the token cost, so the total stays below the baseline, a net saving. For SH-High the technology cost, driven by the GPU infrastructure, absorbs the efficiency saving and pushes the total above the baseline, a net loss.

Overall, in the assisted regime ($\bar\eta>0$) adoption is profitable with confidence only at low capex, and the capex that adoption can bear responds far more to automation ($\bar\eta\to 0$) than to the error cost ($\bar\beta\to 0$); Sec.~\ref{sec:res_sensitivity} ranks both against the remaining drivers.

\subsection{Downside risk}\label{sec:res_var}

\begin{figure*}
    \centering
    \includegraphics[width=1\linewidth]{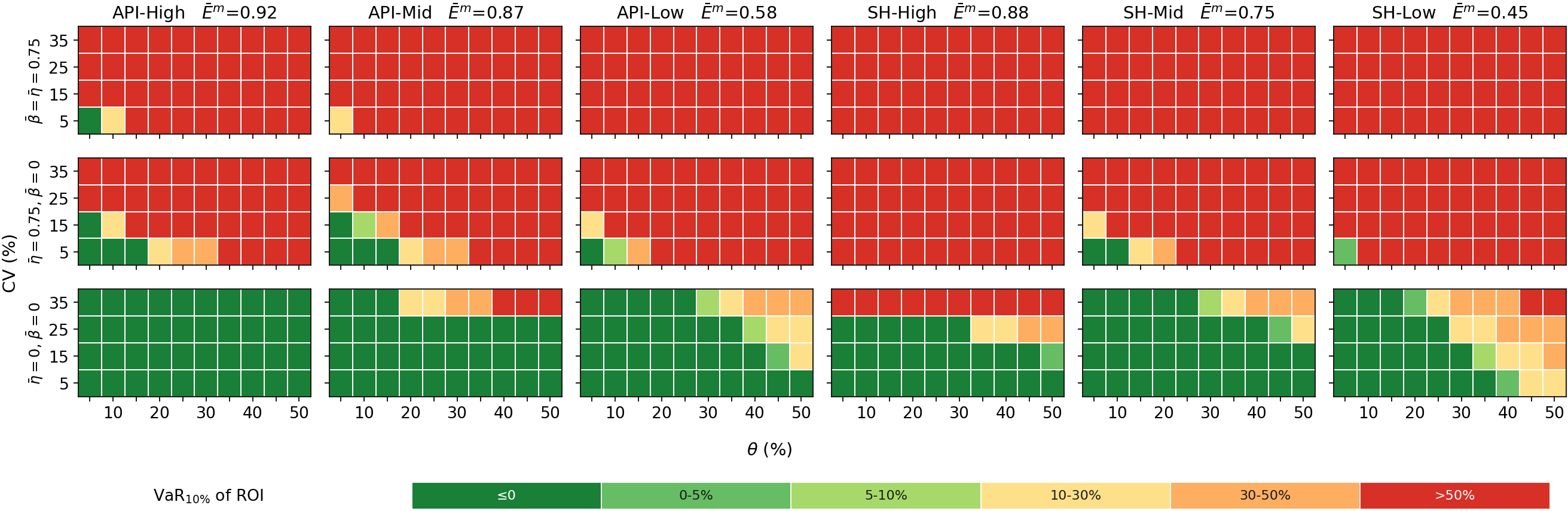}
    \caption{Downside risk of adoption $\mathrm{VaR}_{10\%}(\mathrm{ROI})$ (worst-decile loss), as $\theta$ and CV vary. Rows: mode $(\bar\eta,\bar\beta)$. Col.: archetypes $(\bar{E}^m)$.}
    \label{fig:var10}
\end{figure*}

Figure~\ref{fig:var10} turns from the probability of profit to its downside: $\mathrm{VaR}_{10\%}$ of the ROI, the fraction of the committed capex lost in the worst decile of outcomes (with $10\%$ probability). Green cells ($\mathrm{VaR}_{10\%}\le 0$) mark where even that worst decile is still profitable, red cells a severe loss. 
 
The scenario ordering of Fig.~\ref{fig:proi} carries over to the downside: the safe (green) region widens from the conservative scenario ($\bar\beta=\bar\eta$), through $\bar\beta=0$, to full automation ($\bar\eta=0$), and shrinks as $\theta$ or CV grows. The $\mathrm{VaR}_{10\%}$ ranking tracks profitability when API and self-hosted archetypes are compared separately. Within the self-hosted archetypes, the GPU cost outweighs the accuracy gain, so the larger SH-High is less safe than SH-Mid.

Beyond the safe region, the graded bands give the size of the loss: in the conservative scenario the worst decile loses more than half the committed capex across almost the whole grid. Under full automation, the worst decile stays profitable up to high $\theta$ over most of the grid. For a risk-averse SOC this sets the capex ceiling: adoption is defensible only where the worst-decile loss stays below the loss the organization is willing to bear.

\subsection{Parameter sensitivity}\label{sec:res_sensitivity}

\begin{figure}[pos=b]
    \centering
    \includegraphics[width=1\linewidth]{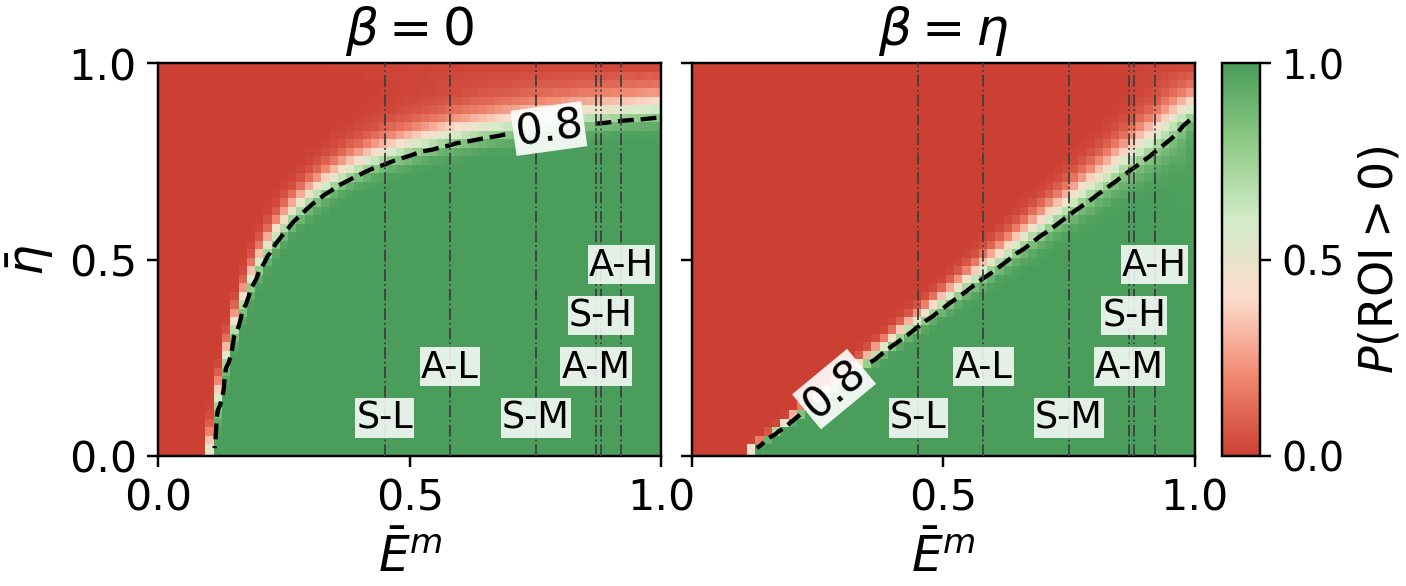}
    \caption{Profit region in the $(\bar{E}^m,\bar\eta)$ plane; $C^m_\text{tech}=0$, $\mathrm{CV}=5\%$, $\bar\theta=10\%$. Dash-dot lines: archetypes, coded by deployment (A: API, S: self-hosted) and accuracy tier (H/M/L).}
    \label{fig:profit_reg}
\end{figure}

In this section, we quantify the impact of each parameter on the value an LLM deployment creates. As a first step, we set the technology cost to zero ($C^m_\text{tech}=0$) to isolate the operational core of the decision: accuracy $\bar{E}^m$ and speed ratio $\bar\eta$. Figure~\ref{fig:profit_reg} maps $P(\mathrm{ROI}>0)$ over that plane at $\mathrm{CV}=5\%$ and $\bar\theta=10\%$, showing when the two leave a positive return. A deployment is profitable only above an $\bar{E}^m$ threshold (dashed line) that rises with $\bar\eta$, since effort left to the analyst has to be paid back by accuracy. The rework penalty $\beta$ sets where this threshold lies: in relation to the archetypes' accuracies (vertical dash-dot lines), at high rework ($\beta=\bar\eta$) the admissible $\bar\eta$ spreads across the full range, whereas at $\beta=0$ these ceilings collapse into a narrow $[0.7$--$0.8]$ band.

Restoring the technology cost $C^m_\text{tech}$, we run a sensitivity analysis measuring how much each operational parameter shifts the baseline ROI set by the calibration values (Fig.~\ref{fig:tornado}).  

For the high- and mid-accuracy models the speed ratio $\bar\eta$ is the dominant lever, while the rework penalty $\bar\beta$ is markedly weaker. For the low-accuracy models $\bar\beta$ becomes comparable to $\bar\eta$, and for SH-Low it is the largest lever. As in Fig.~\ref{fig:roivsN}, volume $\bar N$ does not affect the API return and moves only the self-hosted models, most visibly SH-High.

In a SOC, an incorrect answer that escapes detection and is exploited by an attacker is a breach; this risk grows under full automation, where outputs are not individually reviewed. Eq.~\eqref{eq:breach-budget} then bounds the breach rate adoption can sustain. For example, if we set $C_\text{miss}$ to the average U.S. breach cost of \$10.22 million~\cite{IBM2025}, at mean values ($\bar\eta=0$, $\bar\beta=0$, $\bar\theta=10\%$, CV$=5\%$), API-High gives $1/\lambda_{\max}\approx 7$ years. A SOC should enable full automation only where its breach rate stays below $\lambda_{\max}$.

\begin{figure*}
    \centering
    \includegraphics[width=1\linewidth]{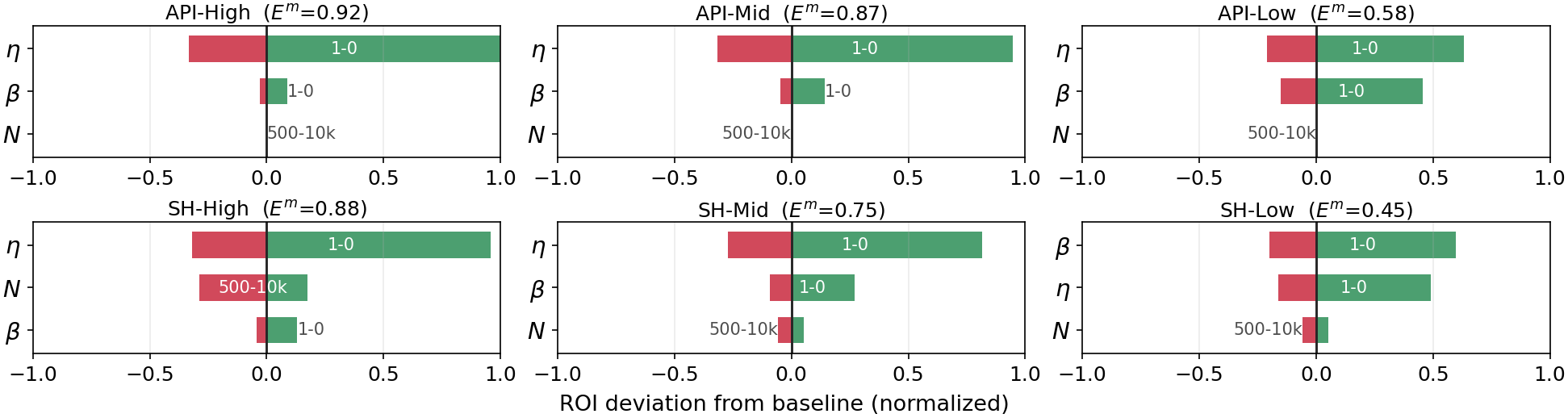}
    \caption{ROI sensitivity analysis per archetype, as deviation from the baseline operating point ($\bar\eta=\bar\beta=0.75$, $\bar\theta=10\%$, $\bar N=1000$); shared normalized scale, $0$ = baseline; green above baseline, red below.}
    \label{fig:tornado}
\end{figure*}

\section{Discussion}
Our analysis shows that only a few factors substantially influence the return on LLM adoption. These factors represent three operational levers that decision-makers can control:
\begin{itemize}[leftmargin=*]
    \item \textbf{Model choice} ($E^m$, accuracy): the selected model's correctness on the target task, including its tendency to hallucinate. This characteristic is determined when the model is chosen rather than tuned during operation.
    \item \textbf{Level of automation} ($\eta$, speed ratio): the share of analyst time still required for the assisted task. A lower value means that the model operates more independently, potentially reaching fully agentic execution.
    \item \textbf{Level of verifiability} ($\beta$, rework penalty): how readily an analyst recognizes an incorrect output before acting on it, rather than redoing the task from scratch. More verifiable outputs make errors less costly.
\end{itemize}

Our simulations identify automation as the main lever for turning a loss into a profit. Returns are far more sensitive to the speed ratio $\eta$ than to the rework penalty $\beta$ (Fig.~\ref{fig:tornado}). The autonomy regime ($\bar\eta\to0$) is the only one that can accommodate high capital expenditure while remaining profitable even in its worst-decile outcomes. However, this level of autonomy requires sufficient accuracy: only a model precise enough to be trusted without full human re-verification can be operated this way (Fig.~\ref{fig:profit_reg}). At the calibrated operating point, where this degree of autonomy is absent, most configurations do not pay off. Under the conservative regime ($\bar\beta=\bar\eta$), only the high-accuracy API archetypes achieve high-confidence profitability, and then only at low capital expenditure $\theta$ and low uncertainty $\mathrm{CV}$, while every other configuration falls short of it (Fig.~\ref{fig:proi}, Fig.~\ref{fig:var10}). This finding is consistent with reports that enterprise generative AI often fails to convert pilots into value~\cite{MIT2025, SPGlobal2025, McKinsey2024}. The importance of autonomy also aligns with market trends toward agentic, human-in-the-loop security operations and model specialization~\cite{Gartner2026}. Adoption therefore pays off only as security operations move toward the cognitive-SOC regime, where the model handles most of the work and analysts address the remainder.

Price has only a marginal effect in our simulations. For API deployments, token costs account for a small share of the total (Fig.~\ref{fig:cost}). Moreover, price and accuracy are uncorrelated across the surveyed models (Spearman $\rho=0.04$, Sec.~\ref{sec:archetypes}), suggesting that a model's price reflects its market positioning more than its capabilities.

Alert volume $N$ has little effect on returns from API deployments because costs and savings scale together, leaving the ratio unchanged. By contrast, self-hosted deployment rarely helps. At realistic alert volumes, fixed GPU infrastructure costs outweigh the efficiency savings and push total costs above the no-LLM baseline (Fig.~\ref{fig:cost}). These costs are amortized only at volumes well above those of a typical SOC (Fig.~\ref{fig:roivsN}). The API route therefore remains more convenient for most SOCs, with $N$ serving as an indicator of SOC scale because larger SOCs receive more alerts. However, privacy requirements that prohibit sending data to the cloud, together with emerging models that reduce deployment requirements without sacrificing performance, could shift the balance toward self-hosting.

The setting is also shaped by the coefficient of variation, $\mathrm{CV}$. Its uncertainty reflects both limited knowledge of appropriate input values and their intrinsic variability. For example, a mature use case with known, stable values, such as documentation or reporting, has a low $\mathrm{CV}$. A high $\mathrm{CV}$ instead characterizes novel or poorly measured settings, such as a SOC that does not yet track its metrics adequately or a domain where LLM use is not established. Sweeping $\mathrm{CV}$ therefore shows how the same decision performs across established, well-understood and more speculative settings. Accordingly, we assess uncertainty directly through ROI and VaR. In analyses that isolate a single factor, we fix $\mathrm{CV}$ at its low level to limit its influence on the results.

The model is designed to support decision-making. Consider, for example, a mid-to-large SOC evaluating an LLM rollout. The decision-maker, such as the SOC manager, sets the operating point using the organization's own data, including analyst costs, effort profiles, and alert distributions. Most of these are already tracked as standard SOC metrics (Sec.~\ref{sec:reference-soc}), and estimates become more precise as more inputs are measured. After selecting a candidate model, whose accuracy and price are then known, the decision-maker can examine $P(\mathrm{ROI}>0)$ across capital expenditure ratios $\theta$ at the SOC's uncertainty level (Fig.~\ref{fig:proi}). This indicates the highest capex at which adoption meets the chosen confidence threshold. The $\mathrm{VaR}_{10\%}$ (Fig.~\ref{fig:var10}) indicates the potential loss in the worst decile. Adoption is defensible when the capex supported at the chosen confidence level covers the estimated one-off cost and the worst-decile loss remains within tolerance.

We release the simulation code so that practitioners can use their own parameters and boundary conditions.

\subsection{Limitations and future works}\label{sec:limitations}
The model and its SOC instantiation provide a first, deliberately conservative demonstration that LLM adoption can be assessed economically before an organization commits to it. Several extensions could strengthen and broaden the analysis, as outlined below.
\begin{itemize}[leftmargin=*]
\item \textit{Empirical validation.} We instantiated the model using public benchmarks and the authors' operational experience, which highlighted how difficult it is to obtain reliable data. A natural next step is to validate the model using operational data from a live SOC. This would test its fit to real cases and assess how consistently the required metrics are tracked.

\item \textit{Model refinement.} To remain applicable across domains, the model makes several assumptions, including the use of aggregate values. Future versions could distinguish accuracy, speed ratio, and rework penalty by task, role, and severity. They could also introduce dedicated AI roles and tasks, which the current model subsumes into existing ones. This generality also allows the model to be extended beyond SOCs and cybersecurity to other information-intensive operations.

\item \textit{Infrastructure modeling.} The estimated self-hosted cost is a conservative upper bound because it does not account for serving efficiencies in real deployments, such as request batching and KV-cache reuse. This area is evolving rapidly, and such efficiencies may be large enough to change the adoption decision. A dedicated study of real self-hosted deployments would therefore provide a more accurate comparison with the API option.

\item \textit{Security-specific benchmarks.} Deriving accuracy values $E^m$ from general cyber benchmarks revealed an important limitation. These heterogeneous benchmarks cover largely disjoint sets of models and subdomains, and report different accuracy scores for the same model, making them imperfect proxies. A rigorous benchmark of model efficacy on security tasks that is both scenario-specific and comparable is still lacking. Such a benchmark would provide a firmer basis for the accuracy input and support broader evaluations of LLM performance in the security domain.

\end{itemize}

\section{Conclusion}\label{sec:conclusion}
Whether adopting an LLM pays off in an operational setting depends on a few operational factors, as well as on the choice of model. We framed the decision as an investment under deep uncertainty and formalized it in an economic framework, which we instantiated on a Security Operations Center as a first application. Rather than reporting a single ROI figure, the framework provides the full ROI distribution and a clear basis for a go/no-go decision, including the potential loss if the decision goes wrong.

At the simulated operating points, most configurations are not reliably profitable. Only high-accuracy models are profitable, and only when both capital costs and uncertainty are low. Outcomes depend less on model price than on how far its output can be acted on without human review. Paying more does not yield a better return, whereas granting the model greater autonomy does, provided it is accurate enough to be trusted. Self-hosting rarely helps at a typical SOC's alert volume, where GPU costs outweigh the savings.

Using the open-source code we release, an organization can enter its own figures and assess whether adoption pays off, at what capital cost, and with what risk of loss. As security operations move toward the agentic cognitive SOC model, the same tool shows when the economics become favorable. Because the framework is based on generic events and criticality, it is also designed to extend beyond the SOC to other information-intensive operations. Adoption is not yet a winning bet, but the conditions under which it becomes one are now measurable.

\section*{Funding}
This research did not receive any specific grant from funding agencies in the public, commercial, or not-for-profit sectors.

\section*{Data availability}
The Monte Carlo simulation code and the benchmark-derived parameters used in this study will be made publicly available upon publication. All source benchmarks are cited in the paper.

\section*{Declaration of generative AI and AI-assisted technologies in the manuscript preparation process}
During the preparation of this work, the authors used Claude (Anthropic) for language editing, structural revision, clarity improvements, and coding support. After using these tools, the authors reviewed and edited all content and code, and take full responsibility for the published article. 

\section*{Acknowledgments}
This work was partly carried out during a research visit of the first author to Télécom SudParis, Institut Polytechnique de Paris, as part of his PhD program. The authors also thank the security practitioners who provided operational insights that informed the model's parameterization.

\printcredits

\bibliographystyle{unsrtnat}
\bibliography{BIBLIO_DSS_hypeROI_v1_2}

\appendix

\section{Proofs and Distribution Parametrization}\label{app:proofs}

\subsection{Return-on-investment threshold}\label{app:ROI-proof}
Proposition~\ref{prop:ROI} bounds the capital expenditure ratio $\theta$ compatible with a target return $y$. The bound follows by rewriting the ROI as a function of $\theta$.

\begin{proof}[Proof of Proposition~\ref{prop:ROI}]
From \eqref{eq:theta} and \eqref{eq:ROI},
\[
\textit{ROI}=\frac{\Delta C-C_\text{cap}}{C_\text{cap}}
=\frac{\Delta C}{C_\text{cap}}-1
=\frac{\Delta C}{\theta\,C\nollm}-1 .
\]
Therefore
\[
\textit{ROI}>y \iff \frac{\Delta C}{\theta\,C\nollm}>1+y
\iff \theta<\frac{\Delta C}{(1+y)\,C\nollm}. \qedhere
\]
\end{proof}

\subsection{Distribution parametrization}\label{app:par-proof}
Each uncertain parameter is specified by its mean $\mu$ and coefficient of variation $\mathrm{CV}$; the two lemmas below recover, for the Beta and Gamma families, the native distribution parameters from $\mu$ and $\mathrm{CV}$.

\begin{lemma}
\label{lem:beta}
Let $X$ be a beta-distribution with parameter $(\alpha,\beta)$ with $\alpha>0,\beta>0$ and mean $\mu\in(0,1)$.
Then the parameters $(\alpha,\beta)$ can be expressed as
\[
\alpha \;=\; \mu\left(\frac{1-\mu}{\mu\,\mathrm{CV}^{2}}-1\right),
\qquad
\beta \;=\; (1-\mu)\left(\frac{1-\mu}{\mu\,\mathrm{CV}^{2}}-1\right).
\]
\end{lemma}
\begin{proof}
For $X\sim \mathrm{Beta}(\alpha,\beta)$, the mean and variance are
\[
\mu=\; \frac{\alpha}{\alpha+\beta},
\qquad
\mathrm{Var}(X) \;=\; \frac{\alpha\beta}{(\alpha+\beta)^2(\alpha+\beta+1)}.
\]
Let $s \triangleq \alpha+\beta$. We have
\begin{align}
\label{eq:alpha-beta}
\alpha=\mu s,
\qquad
\beta=s-\alpha=(1-\mu)s.
\end{align}
Substituting into the variance yields
\[
\mathrm{Var}(X)
=\frac{(\mu s)\bigl((1-\mu)s\bigr)}{s^2(s+1)}
=\frac{\mu(1-\mu)}{s+1}.
\]
By definition of the coefficient of variation,
\[
\mathrm{CV}^2
=\frac{\mathrm{Var}(X)}{\mu^2}
=\frac{\mu(1-\mu)}{(s+1)\mu^2}
=\frac{1-\mu}{\mu(s+1)},
\]
and thus $s=\frac{1-\mu}{\mu\,\mathrm{CV}^2}-1$. Replacing this into~\eqref{eq:alpha-beta} proves the claim.
\end{proof}

\begin{lemma}
\label{lem:gamma}
Let $X$ be a Gamma-distributed parameter with \emph{shape} parameter $\kappa>0$ and \emph{scale} parameter $\varsigma>0$ and mean $\mu$. Such parameters can be expressed as
\[
\kappa \;=\; 1/\mathrm{CV}^{2},
\qquad
\varsigma \;=\; \mu\,\mathrm{CV}^{2}.
\]
\end{lemma}
\begin{proof}
For a Gamma distribution parameterized by shape $\kappa$ and scale $\varsigma$,
the mean and variance are given by
\[
\mu \;=\; \kappa\,\varsigma,
\qquad
\mathrm{Var}(X) \;=\; \kappa\,\varsigma^{2}.
\]
By definition of the coefficient of variation,
\[
\mathrm{CV}^{2}
=\frac{\mathrm{Var}(X)}{\mu^2}
=\frac{\kappa\,\varsigma^{2}}{(\kappa\,\varsigma)^2}
=\frac{1}{\kappa}
\Rightarrow
\kappa=\frac{1}{\mathrm{CV}^{2}} .
\]
Substituting this expression into $\mu=\kappa\,\varsigma$ gives
$\varsigma=\mu/\kappa=\mu\,\mathrm{CV}^{2}$.
\end{proof}

\onecolumn

\section{Cybersecurity Benchmark-Derived Model Parameters} \label{appendix: model_params}
This section lists the models collected from the six benchmarks~\cite{deason_cybersoceval_2025, jing_secbench_2025, alam_ctibench_2024, tihanyi_cybermetric_2024, cherif_dfir-metric_2025, yu_cs-eval_2025}, with their deployment mode, token prices, total and active parameters, mean benchmark accuracy, and per-model references; prices are the latest published, or the last archived value for deprecated models.

\setlength{\LTcapwidth}{\textwidth}
{\small
\begin{longtable}{l l l r r r r r l}
\caption{LLM pricing, size and cybersecurity benchmark accuracy. 71 distinct models (73 rows): DeepSeek-V3 and DeepSeek-R1 evaluated under both API and self-hosted deployment.}\label{tab:llm-pricing} \\
\toprule
\textbf{Model} & \textbf{Provider} & \textbf{Deploy Mode} & $\boldsymbol{\gamma^{m,\text{in}}}$ & $\boldsymbol{\gamma^{m,\text{out}}}$ & $\boldsymbol{\delta^m_\text{tot}}$ & $\boldsymbol{\delta^m_\text{act}}$ & $\boldsymbol{E^m}$ & \textbf{Ref.} \\
\midrule
\endfirsthead

\multicolumn{9}{c}{\tablename\ \thetable{} -- \textit{continued from previous page}} \\[0.5ex]
\toprule
\textbf{Model} & \textbf{Provider} & \textbf{Deploy Mode} & $\boldsymbol{\gamma^{m,\text{in}}}$ & $\boldsymbol{\gamma^{m,\text{out}}}$ & $\boldsymbol{\delta^m_\text{tot}}$ & $\boldsymbol{\delta^m_\text{act}}$ & $\boldsymbol{E^m}$ & \textbf{Ref.} \\
\midrule
\endhead

\midrule
\multicolumn{9}{r}{\textit{Continued on next page}} \\
\endfoot

\bottomrule
\endlastfoot

    Hunyuan-Turbo & Tencent & API & \$0.33 & \$1.33 & -- & -- & 94\% & \cite{jing_secbench_2025} \\
    gpt-4.1 & OpenAI & API & \$2.00 & \$8.00 & -- & -- & 93\% & \cite{cherif_dfir-metric_2025} \\
    gpt-4-turbo & OpenAI & API & \$10.00 & \$30.00 & -- & -- & 92\% & \cite{tihanyi_cybermetric_2024} \\
    Grok-3 & xAI & API & \$3.00 & \$15.00 & -- & -- & 92\% & \cite{cherif_dfir-metric_2025} \\
    DeepSeek-V3 & DeepSeek & API & \$0.27 & \$1.10 & -- & -- & 91\% & \cite{jing_secbench_2025,cherif_dfir-metric_2025} \\
    o1-mini & OpenAI & API & \$1.10 & \$4.40 & -- & -- & 91\% & \cite{jing_secbench_2025} \\
    gemini-2.5-flash & Google & API & \$0.30 & \$2.50 & -- & -- & 90\% & \cite{cherif_dfir-metric_2025} \\
    o1-preview & OpenAI & API & \$15.00 & \$60.00 & -- & -- & 90\% & \cite{jing_secbench_2025} \\
    gpt-4o-mini & OpenAI & API & \$0.15 & \$0.60 & -- & -- & 87\% & \cite{jing_secbench_2025,cherif_dfir-metric_2025} \\
    gemini-1.0-pro & Google & API & \$0.50 & \$1.50 & -- & -- & 87\% & \cite{tihanyi_cybermetric_2024} \\
    gpt-3.5-turbo-16k & OpenAI & API & \$3.00 & \$4.00 & -- & -- & 81\% & \cite{yu_cs-eval_2025} \\
    gpt-4o & OpenAI & API & \$2.50 & \$10.00 & -- & -- & 79\% & \cite{deason_cybersoceval_2025,jing_secbench_2025,tihanyi_cybermetric_2024,cherif_dfir-metric_2025,yu_cs-eval_2025} \\
    gpt-3.5-turbo & OpenAI & API & \$0.50 & \$1.50 & -- & -- & 75\% & \cite{jing_secbench_2025,alam_ctibench_2024,tihanyi_cybermetric_2024} \\
    gpt-4 & OpenAI & API & \$30.00 & \$60.00 & -- & -- & 71\% & \cite{alam_ctibench_2024,yu_cs-eval_2025} \\
    claude-3-7-sonnet & Anthropic & API & \$3.00 & \$15.00 & -- & -- & 58\% & \cite{deason_cybersoceval_2025,cherif_dfir-metric_2025} \\
    gemini-1.5-pro & Google & API & \$1.25 & \$5.00 & -- & -- & 57\% & \cite{alam_ctibench_2024} \\
    o3 & OpenAI & API & \$2.00 & \$8.00 & -- & -- & 41\% & \cite{deason_cybersoceval_2025} \\
    gemini-2.5-pro & Google & API & \$1.25 & \$10.00 & -- & -- & 36\% & \cite{deason_cybersoceval_2025} \\
    DeepSeek-R1 & DeepSeek & API & \$0.55 & \$2.19 & -- & -- & 31\% & \cite{deason_cybersoceval_2025} \\
    Qwen-2-72B & Alibaba & self-hosted & -- & -- & 72 & 72 & 92\% & \cite{jing_secbench_2025} \\
    DeepSeek-V3 & DeepSeek & self-hosted & -- & -- & 671 & 37 & 91\% & \cite{jing_secbench_2025,cherif_dfir-metric_2025} \\
    Mixtral-8x7B-Instruct & Mistral AI & self-hosted & -- & -- & 46.7 & 12.9 & 91\% & \cite{tihanyi_cybermetric_2024} \\
    Qwen-2.5-72B & Alibaba & self-hosted & -- & -- & 72 & 72 & 90\% & \cite{cherif_dfir-metric_2025} \\
    Yi-1.5-34B & 01.AI & self-hosted & -- & -- & 34 & 34 & 90\% & \cite{jing_secbench_2025} \\
    Falcon-180B-Chat & TII & self-hosted & -- & -- & 180 & 180 & 88\% & \cite{tihanyi_cybermetric_2024} \\
    Qwen-2-7B & Alibaba & self-hosted & -- & -- & 7.6 & 7.6 & 88\% & \cite{jing_secbench_2025} \\
    Qwen-2-72B-Instruct & Alibaba & self-hosted & -- & -- & 72 & 72 & 87\% & \cite{yu_cs-eval_2025} \\
    llama-3.3-70b & Meta & self-hosted & -- & -- & 70 & 70 & 86\% & \cite{cherif_dfir-metric_2025} \\
    Yi-1.5-9B & 01.AI & self-hosted & -- & -- & 9 & 9 & 86\% & \cite{jing_secbench_2025} \\
    DeepSeek-V2 & DeepSeek & self-hosted & -- & -- & 236 & 21 & 85\% & \cite{yu_cs-eval_2025} \\
    GLM-4-9B & Zhipu AI & self-hosted & -- & -- & 9 & 9 & 85\% & \cite{jing_secbench_2025,yu_cs-eval_2025} \\
    gemma-3-27b & Google & self-hosted & -- & -- & 27 & 27 & 85\% & \cite{cherif_dfir-metric_2025} \\
    WizardLM-2 & Microsoft & self-hosted & -- & -- & 141 & 39 & 85\% & \cite{cherif_dfir-metric_2025} \\
    llama-3.1-70b-instruct & Meta & self-hosted & -- & -- & 70 & 70 & 84\% & \cite{yu_cs-eval_2025} \\
    Mixtral-8x7B & Mistral AI & self-hosted & -- & -- & 46.7 & 12.9 & 83\% & \cite{jing_secbench_2025,cherif_dfir-metric_2025} \\
    Qwen-2-7B-Instruct & Alibaba & self-hosted & -- & -- & 7.6 & 7.6 & 83\% & \cite{yu_cs-eval_2025} \\
    gemma-2-9b & Google & self-hosted & -- & -- & 9.2 & 9.2 & 80\% & \cite{cherif_dfir-metric_2025} \\
    DeepSeek-V2-Lite & DeepSeek & self-hosted & -- & -- & 16 & 2.4 & 79\% & \cite{jing_secbench_2025} \\
    Qwen-14B-Chat & Alibaba & self-hosted & -- & -- & 14 & 14 & 79\% & \cite{yu_cs-eval_2025} \\
    llama-3.1-8b-instruct & Meta & self-hosted & -- & -- & 8 & 8 & 77\% & \cite{yu_cs-eval_2025} \\
    gemma-1.1-7b-it & Google & self-hosted & -- & -- & 8.5 & 8.5 & 77\% & \cite{tihanyi_cybermetric_2024} \\
    Qwen-1.5-14B-Chat & Alibaba & self-hosted & -- & -- & 14 & 14 & 77\% & \cite{yu_cs-eval_2025} \\
    llama-3-8b-instruct & Meta & self-hosted & -- & -- & 8 & 8 & 76\% & \cite{tihanyi_cybermetric_2024} \\
    Mistral-7B-Instruct-v0.2 & Mistral AI & self-hosted & -- & -- & 7 & 7 & 75\% & \cite{tihanyi_cybermetric_2024,yu_cs-eval_2025} \\
    Qwen-1.5-MoE-A2.7B-Chat & Alibaba & self-hosted & -- & -- & 14 & 2.7 & 75\% & \cite{yu_cs-eval_2025} \\
    Baichuan-2-13B-Chat & Baichuan & self-hosted & -- & -- & 13 & 13 & 74\% & \cite{yu_cs-eval_2025} \\
    zephyr-7b-beta & HuggingFace & self-hosted & -- & -- & 7 & 7 & 74\% & \cite{tihanyi_cybermetric_2024} \\
    flan-t5-xxl & Google & self-hosted & -- & -- & 11 & 11 & 72\% & \cite{tihanyi_cybermetric_2024} \\
    llama-2-70b & Meta & self-hosted & -- & -- & 70 & 70 & 72\% & \cite{tihanyi_cybermetric_2024} \\
    Baichuan-13B-Chat & Baichuan & self-hosted & -- & -- & 13 & 13 & 71\% & \cite{yu_cs-eval_2025} \\
    llama-3-70b & Meta & self-hosted & -- & -- & 70 & 70 & 68\% & \cite{jing_secbench_2025,alam_ctibench_2024} \\
    360Zhinao-7B-Chat-4K & Qihoo 360 & self-hosted & -- & -- & 7 & 7 & 66\% & \cite{yu_cs-eval_2025} \\
    Yi-6B-Chat & 01.AI & self-hosted & -- & -- & 6 & 6 & 65\% & \cite{yu_cs-eval_2025} \\
    Qwen-1.5-7B & Alibaba & self-hosted & -- & -- & 7.6 & 7.6 & 64\% & \cite{tihanyi_cybermetric_2024} \\
    Qwen-1.5-MoE-A2.7B & Alibaba & self-hosted & -- & -- & 14 & 2.7 & 62\% & \cite{tihanyi_cybermetric_2024} \\
    chatglm3-6b & Zhipu AI & self-hosted & -- & -- & 6 & 6 & 57\% & \cite{yu_cs-eval_2025} \\
    Mistral-3B & Mistral AI & self-hosted & -- & -- & 3 & 3 & 56\% & \cite{cherif_dfir-metric_2025} \\
    Qwen-7B & Alibaba & self-hosted & -- & -- & 7.6 & 7.6 & 53\% & \cite{tihanyi_cybermetric_2024} \\
    Phi-2 & Microsoft & self-hosted & -- & -- & 2.7 & 2.7 & 52\% & \cite{tihanyi_cybermetric_2024} \\
    llama-3-ChatQA-1.5-8b & NVIDIA & self-hosted & -- & -- & 8 & 8 & 51\% & \cite{tihanyi_cybermetric_2024} \\
    DeciLM-7B & Deci AI & self-hosted & -- & -- & 7 & 7 & 50\% & \cite{tihanyi_cybermetric_2024} \\
    SecGPT-13B & Clouditera & self-hosted & -- & -- & 13 & 13 & 47\% & \cite{yu_cs-eval_2025} \\
    llama-3-8b & Meta & self-hosted & -- & -- & 8 & 8 & 45\% & \cite{jing_secbench_2025,alam_ctibench_2024,tihanyi_cybermetric_2024} \\
    Qwen-1.5-4B & Alibaba & self-hosted & -- & -- & 4 & 4 & 40\% & \cite{tihanyi_cybermetric_2024} \\
    llama-4-Maverick-17b-128E-instruct & Meta & self-hosted & -- & -- & 400 & 17 & 39\% & \cite{deason_cybersoceval_2025} \\
    Genstruct-7B & NousResearch & self-hosted & -- & -- & 7 & 7 & 38\% & \cite{tihanyi_cybermetric_2024} \\
    llama-2-13b-chat-hf & Meta & self-hosted & -- & -- & 13 & 13 & 38\% & \cite{yu_cs-eval_2025} \\
    gemma-7b & Google & self-hosted & -- & -- & 8.5 & 8.5 & 37\% & \cite{tihanyi_cybermetric_2024} \\
    llama-4-Scout-17b-16E & Meta & self-hosted & -- & -- & 109 & 17 & 33\% & \cite{deason_cybersoceval_2025} \\
    DeepSeek-R1 & DeepSeek & self-hosted & -- & -- & 671 & 37 & 31\% & \cite{deason_cybersoceval_2025} \\
    Dolly-V2-12b-BF16 & Databricks & self-hosted & -- & -- & 12 & 12 & 30\% & \cite{tihanyi_cybermetric_2024} \\
    gemma-2b & Google & self-hosted & -- & -- & 2.6 & 2.6 & 21\% & \cite{tihanyi_cybermetric_2024} \\
    Phi-3-mini-4k-instruct & Microsoft & self-hosted & -- & -- & 3.8 & 3.8 & 5\% & \cite{tihanyi_cybermetric_2024} \\

\end{longtable}}

\twocolumn

\end{document}